\documentclass[prodmode,acmtkdd]{acmsmall}
\usepackage[english]{babel}
\usepackage[latin1]{inputenc}
\usepackage{amssymb}
\usepackage[ruled]{algorithm2e}
\usepackage{algpseudocode}
\usepackage{bm}
\usepackage{amsmath}
\usepackage{graphicx}
\usepackage{epsfig}
\usepackage{url}
\usepackage{psfrag}
\usepackage{balance}
\usepackage{pstricks}
\usepackage{times}
\usepackage{array}
\usepackage{subfigure}
\usepackage{float}
\usepackage{multirow}

\newcommand {\bpi}{\mbox{\boldmath $\pi$}}
\newcommand {\bphi}{\mbox{\boldmath $\phi$}}

\newcommand {\bP}{\mbox{\boldmath $P$}}

\SetAlFnt{\small}
\SetAlCapFnt{\small}
\SetAlCapNameFnt{\small}
\SetAlCapHSkip{0pt}
\IncMargin{-\parindent}

\begin{document}

\markboth{P. Wang et al.}{Efficiently Estimating Motif Statistics of Large Networks}

\title{Efficiently Estimating Motif Statistics of Large Networks}
\author{Pinghui Wang
\affil{Huawei Noah's Ark Lab}
John C.S. Lui
\affil{The Chinese University of Hong Kong}
Bruno Ribeiro
\affil{Carnegie Mellon University}
Don Towsley
\affil{University of Massachusetts Amherst}
Junzhou Zhao
\affil{Xi'an Jiaotong University}
Xiaohong Guan
\affil{Xi'an Jiaotong University}}

\begin{abstract}
Exploring statistics of locally
connected subgraph patterns (also known as network motifs)
has helped researchers better understand the structure and function of
biological and online social networks (OSNs).
Nowadays the massive size of some critical networks -- often
stored in already overloaded relational databases -- effectively limits
the rate at which nodes and edges can be explored, making it
a challenge to accurately discover subgraph statistics.
In this work, we propose {\em sampling methods} to accurately
estimate subgraph statistics from as few queried nodes as possible.
We present sampling algorithms that efficiently and accurately
estimate subgraph properties of massive networks.
Our algorithms require no pre-computation or
complete network topology information. At the same time,
we provide theoretical guarantees of convergence.
We perform experiments using widely known data sets,
and show that for the same accuracy, our algorithms require an order of magnitude
less queries (samples) than the current state-of-the-art algorithms.
\end{abstract}




\acmformat{Pinghui Wang, John C.S. Lui, Bruno Ribeiro, Don Towsley, Junzhou Zhao, and Xiaohong Guan, 2014.
Efficiently Estimating Motif Statistics of Large Networks.}

\begin{bottomstuff}
This work was supported by the NSF grant CNS-1065133, ARL Cooperative Agreement W911NF-09-2-0053, and ARO under MURI W911NF-08-1-0233. The views and conclusions contained in this document are those of the authors and should not be interpreted as representing the official policies, either expressed or implied of the NSF, ARL, or the U.S. Government. This work was also supported in part by the NSFC funding 60921003 and 863 Program 2012AA011003 of China.

Author's addresses: Pinghui Wang, HUAWEI Noah's Ark lab in Hong Kong. A part of the research was done when he was at Department of Computer Science and Engineering, The Chinese University of Hong Kong.
John C.S. Lui, Department of Computer Science and Engineering, The Chinese University of Hong Kong, Hong Kong; Bruno Ribeiro, School of Computer Science, Carnegie Mellon University, PA, US;
Don Towsley, Department of Computer Science, University of Massachusetts Amherst, MA, US;
Junzhou Zhao {and} Xiaohong Guan, MOE Key Laboratory for Intelligent Networks and Network Security, Xi'an Jiaotong University, Shaanxi, China.
\end{bottomstuff}

\maketitle

\section{Introduction} \label{sec:introduction}
Understanding the structure and function of complex systems is of wide
interest across many fields of science and technology,
from sociology to physics and biology.
Networks with similar topological features such as
degree distribution or graph diameter can exhibit significantly
different local structure. Thus, there is much interest in exploring
small connected subgraph patterns in networks,
which are often shaped during their growth
and have been used to characterize communication and evolution patterns
in OSNs~\cite{ChunIMC2008,Kunegis2009,ZhaoNetsci2011,Ugander2013}.
For example, simple 3-node subgraph classes such as
``the enemy of my enemy is my friend"
and ``the friend of my friend is my friend" are well known evolution patterns in social networks.
Kunegis et al.~\cite{Kunegis2009} considered the significance of
these subgraph patterns in Slashdot Zoo\footnote{www.slashdot.org} evaluating the stability of signed friend/foe subgraphs.
Other examples include  counting relative frequencies of closed triangle (i.e., three users connected with each other), which are probably more prevalent in Facebook than Twitter, since Twitter serves more as a news media than an OSN service~\cite{Kwak2010}.
More complex examples of $k$-node subgraph frequencies, $k > 3$, include
Milo et al.~\cite{Milo2002} defined network motifs (or local
subgraph patterns) as small subgraph classes occurring in networks at numbers
that are significantly larger than found in random networks.
Network motifs have been used for pattern recognition in gene expression
profiling~\cite{Shenorr2002}, protein-protein interaction
predication~\cite{Albert2004}, and coarse-grained topology
generation~\cite{Itzkovitz2005}.

Unfortunately, characterizing the frequencies of subgraph patterns by searching and
counting subgraphs is computationally intensive since the number of
possible $k$-node combinations in the original graph increases
exponentially with $k$.
To address this problem, Kashtan et al.~\cite{Kashtan2004} propose to sample subgraphs using random edge sampling but this method scales poorly with the subgraph size and the results can be heavily biased.
Wernicke~\cite{Wernicke2006} proposes another approach (FANMOD) based on enumerating subgraph trees.
The latter relies on random node sampling, which is either not supported by most OSN APIs or is too resource intensive to be practical (with respect to cache misses and vacant user ID space~\cite{RibeiroINF12}).

\begin{figure*}[htb]
\center
\subfigure[$G$]{
\includegraphics[width=0.16\textwidth]{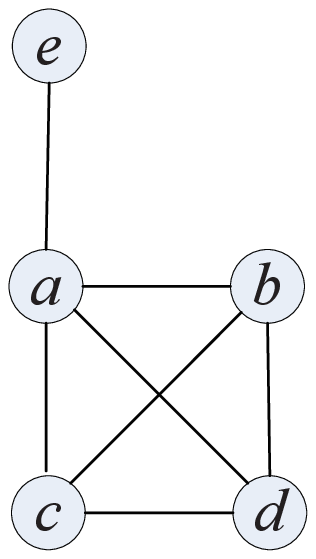}}
\subfigure[$G^{(2)}$]{
\includegraphics[width=0.35\textwidth]{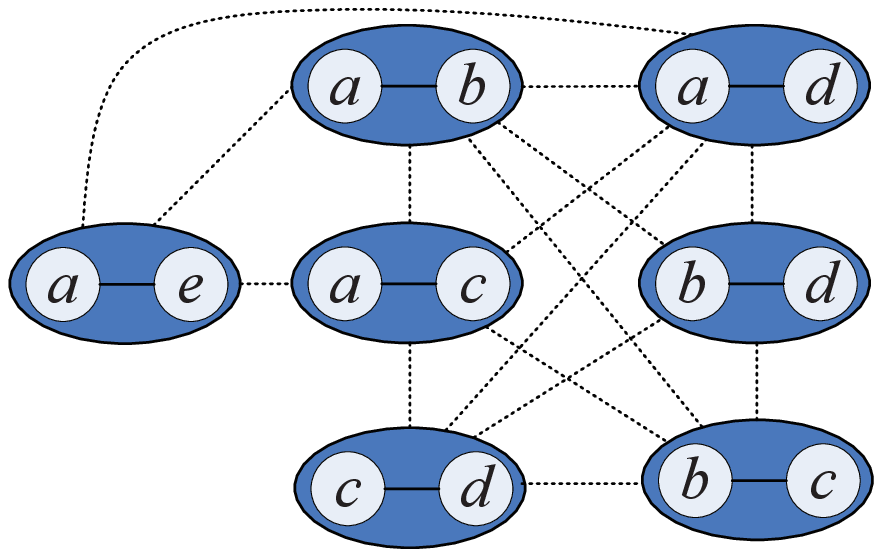}}
\subfigure[$G^{(3)}$]{
\includegraphics[width=0.5\textwidth]{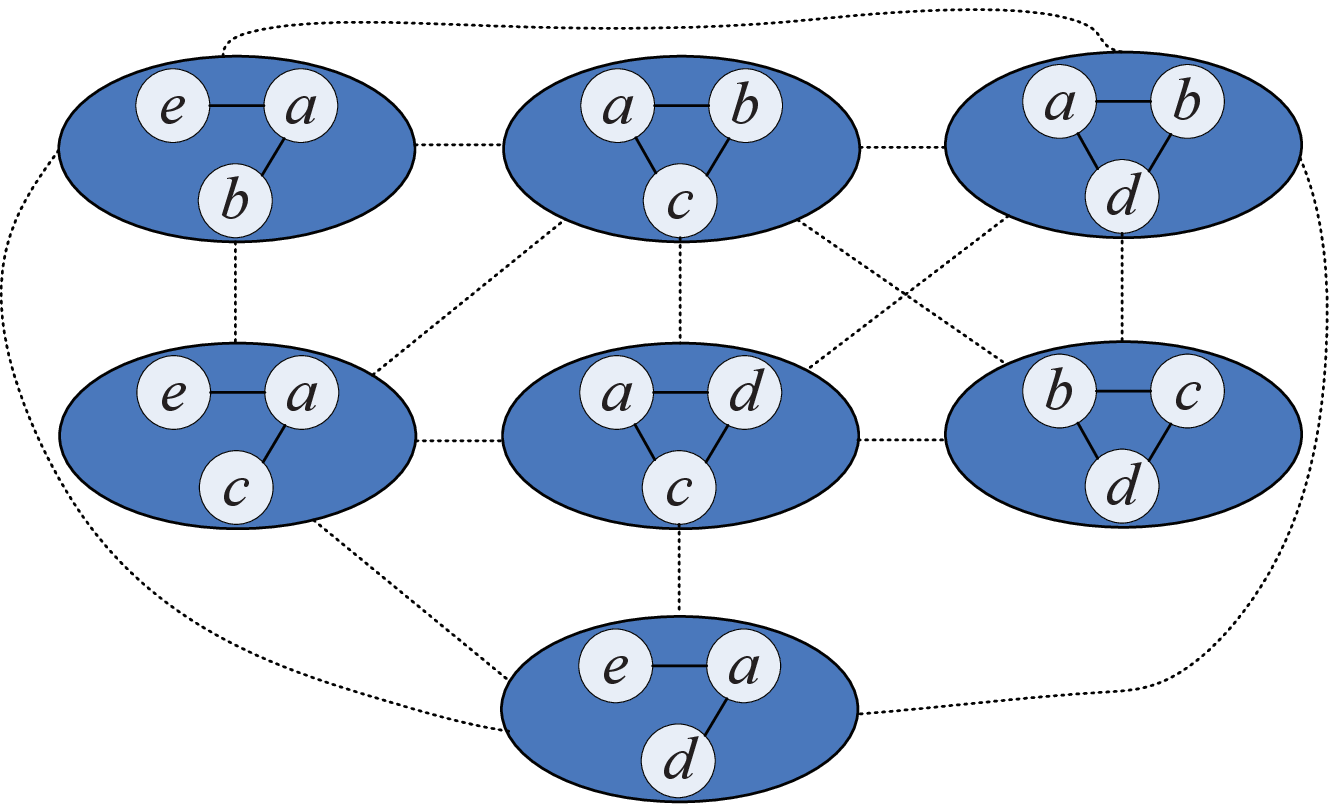}}
\subfigure[$G^{(4)}$]{
\includegraphics[width=0.26\textwidth]{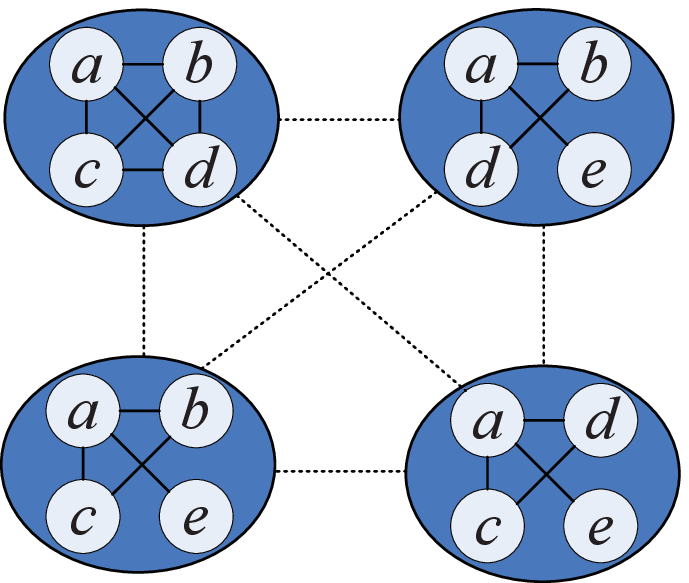}}
\caption{An example of $G$, and CIS relationship graphs $G^{(2)}$, $G^{(3)}$, and $G^{(4)}$.}\label{fig:exampleCISgraph}
\end{figure*}

Thus, it is paramount for such algorithms to query the graph on-the-fly without knowledge of the complete topology.
Moreover, such algorithms should output accurate high accuracy estimates of subgraph concentrations with as few queries as possible.
Recently,  Bhuiyan et al.~\cite{Bhuiyan2012} proposed a Metropolis-Hastings-based algorithm (henceforth denoted GUISE) that jointly estimates concentrations of 3-node, 4-node, and 5-node connected induced subgraphs (CISes).
However, the rejection sampling procedure of the Metropolis-Hastings random walk (MHRW) used in GUISE has received much criticism lately in the precise context of graph sampling~\cite{gjoka2011practical,RibeiroCDC12}.
Rejecting a sample incurs the {\em cost of sampling} without gathering information in exchange.
And, information-wise, the rejected samples may contain more information about the statistic of interest than the accepted samples.
Recently, Ribeiro and Towsley~\cite{RibeiroCDC12} shows that MHRW rejects information-rich samples when trying to estimate the degree distribution of a graph.
The end result is a sampling method that exhibits large estimation errors.
A new estimation method that does not suffer from the above mentioned problems is needed.

In this work we propose two algorithms to accurately estimate subgraph concentrations.
The first one, denoted PSRW, significantly improves upon GUISE in two fronts: (a) PSRW can estimate statistics of CISes of any size, in contrast to GUISE that is limited to jointly estimating 3-node, 4-node, and 5-node CISes; and (b) through careful design PSRW does not reject samples, making the estimation errors of PSRW significantly lower than those of GUISE.
Most importantly, PSRW is not an incremental improvement over GUISE but rather a different type of random walk that is designed without the need to reject samples, using the Horvitz-Thompson estimator~\cite{Ribeiro2010} to unbias the observations.
The second algorithm we propose, denoted Mix Subgraph Sampling or MSS, can jointly estimate CISes of sizes $k-1$, $k$, and $k+1$ for any $k \geq 4$, not only generalizing GUISE (GUISE 3-node, 4-node, and 5-node CISes is the special case $k=4$) but also achieving lower estimation errors.
One of the main differences between GUISE and PSRW or MSS is that our random walk is designed to sample nodes that are important for the CIS estimation and use {\em all of the gathered samples} in the estimation phase.

{\em Through simulations we show that both of our methods (PSRW and MSS) are significantly more accurate than GUISE for either the same number queried nodes or the same walk clock time (using a modern computer). The walk clock time is measured under the assumptions of access to a local database or a remote database (assuming 100 milliseconds of query response delay). }
Our methods represent the network as a {\em CIS relationship graph}, whose nodes are {\em connected and induced subgraphs} (CISes) of the original network.
Fig.~\ref{fig:exampleCISgraph} illustrates a CIS relationship graph for subgraphs of two, three, and four node subgraphs.
Our algorithms consist of running a random walk (RW) on the CIS relationship graph.
Besides its accuracy, our algorithms are lightweight.
They require little memory (more precisely, {\color{black} $O(k^2 + B)$} space where $k$ is the subgraph size and $B$ is the number of queried nodes) and, more importantly, {\em significantly fewer queries than the state-of-the-art methods to achieve the same accuracy}.
Note that building the completely CIS relationship graph is prohibitively expensive, both in terms of queries and memory.
Thus, our RW methods {\em do not require} the CIS relationship graph and there is no need to know the complete graph topology in advance, only the parts of the network already queried.
We also prove that a RW on the CIS relationship graph achieves asymptotically unbiased concentration estimates of the distinct subgraphs on the original network.


This paper is organized as follows. The problem formulation is presented in Section~\ref{sec:problem}. Section~\ref{sec:methods} presents methods for estimating subgraph class concentrations. The performance evaluation and testing results are presented in Section~\ref{sec:results}. Section~\ref{sec:application} presents applications of our methods to two real OSN websites.
Section~\ref{sec:related} summarizes related work. Concluding remarks then follow.

\section{Problem Formulation} \label{sec:problem}
Let $G=(V, E, L)$ be a labeled undirected graph where $V$ is the set of nodes,
$E$ be a set of pairs of $V$ (edges),
and $L$ is a set of labels $l_{i,j}$ associated with edges $(i,j)\in E$.
If $G$ represents a directed network, then we attach a label to each edge that indicates the direction of the edge ($\to$, $\leftarrow$, $\leftrightarrow$).
Edges may have other labels too, for instance, in a signed network, edges have positive or negative labels.

An induced subgraph of $G$, $G'=(V', E', L')$, $V'\subset V$, $E'\subset E$ and $L'\subset L$, is a subgraph whose edges and associated labels are all in $G$, i.e. $E'=\{(i,j): i,j \in V', (i,j)\in E\}$,
$L'=\{l_{i,j}: i,j \in V', (i,j)\in E\}$.
We define $C^{(k)}$ as the set of all connected and induced subgraphs (CISes) with $k$ nodes in $G$.
Then we partition $C^{(k)}$ into $T_k$ equivalence classes $C_1^{(k)}, \ldots, C_{T_k}^{(k)}$ where CISes within $C_i^{(k)}$ are isomorphic and any pair of CISes from $C_i^{(k)}$, $C_j^{(k)}$, $i\ne j$ are not isomorphic.
To illustrate our notation, in what follows we present some simple examples.
Fig.~\ref{fig:34nodeclasses} (a) shows all three-node motifs of any unlabeled undirected network.
When $G$ is an unlabeled undirected network,
then the number of three-node motifs is $T_3=2$,
and $C_1^{(3)}$ and $C_2^{(3)}$ are the sets of CISes in $G$ isomorphic to motifs 1 and 2 in Fig.~\ref{fig:34nodeclasses} (a) respectively.
Fig.~\ref{fig:34nodeclasses} (b) shows all
four-node motifs of any unlabeled undirected network,
in this case $T_4=6$.
Fig.~\ref{fig:3nodeclasses} shows all three-node
motifs of any directed network, in this case $T_3=13$.
Fig.~\ref{fig:signclasses3nodes} shows all motifs when
$G$ is of any signed network, in this case $T_3=7$.

\begin{figure}[htb]
\center
\subfigure[3-node]{
\includegraphics[width=0.18\textwidth]{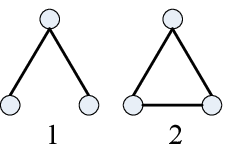}}
\subfigure[4-node]{
\includegraphics[width=0.6\textwidth]{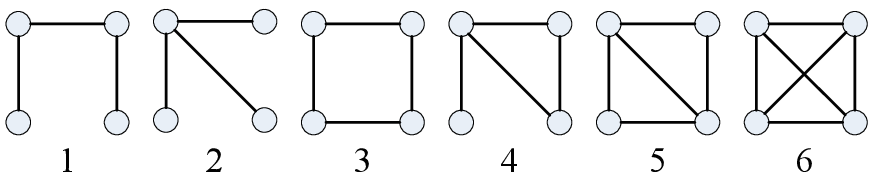}}
\caption{All classes of three-node and four-node undirected and connected subgraphs
         (The numbers are the subgraph class IDs).}
\label{fig:34nodeclasses}
\end{figure}

\begin{figure}[htb]
\begin{center}
\includegraphics[width=0.68\textwidth]{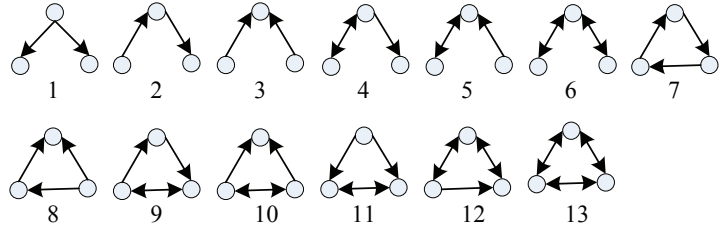}
\caption{All classes of three-node directed and connected subgraphs.}
\label{fig:3nodeclasses}
\end{center}
\end{figure}

\begin{figure}[htb]
\begin{center}
\includegraphics[width=0.65\textwidth]{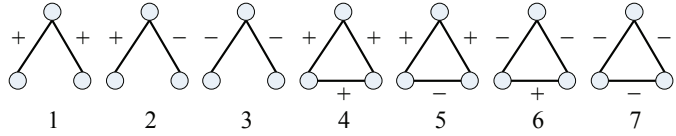}
\caption{All classes of three-node signed and undirected subgraphs.}
\label{fig:signclasses3nodes}
\end{center}
\end{figure}

The concentration of $C_i^{(k)}$ is
\begin{equation*}
\omega_i^{(k)}=\frac{|C_i^{(k)}|}{|C^{(k)}|}, \qquad 1\le i\le T_k,
\end{equation*}
where $|C_i^{(k)}|$ is the number of CISes in $C_i^{(k)}$.
For example, we have $|C^{(3)}|=7$ for $G$ in Fig.~\ref{fig:exampleCISgraph},
and $C_1^{(3)}$  includes three CISes:
1) CIS consisting of $a$, $b$, and $e$;
2) CIS consisting of $a$, $c$, and $e$;
3) CIS consisting of $a$, $d$, and $e$.
Then we have $\omega_1^{(3)}=3/7$.
In this work we are interested in accurately estimating $\omega_i^{(k)}$ by querying a small number of nodes.
Note that the network topology is unknown to us and we are only given one initial connected subgraph of size $k > 1$ in $G$ to bootstrap our algorithm.

\section{\textbf{Connected and Induced Subgraph Sampling Methods}} \label{sec:methods}
In this section we first introduce the notion of a ``{\em CIS relationship graph}".
Then we propose two subgraph sampling methods based on random walks (RWs) on CIS relationship graphs
to estimate the concentrations of subgraph classes of a specific size $k$.
Finally, we propose a sampling method to solve the problem of measuring the concentrations of subgraph classes of sizes $k-1$, $k$, and $k+1$ simultaneously,
where the special case $k=4$ is equivalent to the problem studied in Bhuiyan et al.~\cite{Bhuiyan2012}.
A list of notations used is shown in Table~\ref{tab:notations}.
\begin{table}[htb]
\centering
\tbl{Table of notations}{
\begin{tabular}{||c|l||} \hline
$G=(V, E, L)$&graph under study\\ \hline
$d(v), v\in V$& degree of node $v$ in $G$\\ \hline
$G^{(k)}=(C^{(k)},R^{(k)})$& $k$-node CIS relationship graph\\ \hline
$V(s), s\in C^{(k)}$&set of nodes for the $k$-node CIS $s$ \\ \hline
$E(s), s\in C^{(k)}$&set of edges for the $k$-node CIS $s$\\ \hline
$N(s), s\in C^{(k)}$&$N(s)\subset V$, set of nodes in $V\backslash V(s)$ which are connected to nodes in the $k$-node CIS $s$\\ \hline
$X(s), s\in C^{(k)}$&$X(s)\subset C^{(k)}$, neighbors of $k$-node CIS $s$ in graph $G^{(k)}$\\ \hline
$d^{(k)}(s), s\in C^{(k)}$&  degree of the $k$-node CIS $s$ in $G^{(k)}$\\ \hline
$C(s)$&subgraph class of the CIS $s$\\ \hline
$C_i^{(k)}$& the $i$-th $k$-node subgraph class in $G$\\ \hline
$T_k$&number of $k$-node subgraph classes\\ \hline
$\omega_i^{(k)}$&concentration of subgraph class $C_i^{(k)}$\\ \hline
$I^{(k)}(x), x\in C^{(k+1)}$&number of $k$-node CISes contained in $(k+1)$-node CIS $x$ \\ \hline
$C^{(k-1)}(s)$, $s\in C^{(k)}$ & the set of $(k-1)$-node CISes contained in the CIS $s$ \\ \hline
$O^{(k)}(s')$, $s'\in C^{(k-1)}$ & the set of $k$-node CISes that contain the CIS $s'$ \\ \hline
$B$&number of sampled CISes\\ \hline
$B^*$&number of queries\\ \hline
\end{tabular}}
\label{tab:notations}
\end{table}

\subsection{\textbf{CIS relationship graph}}
Let $C^{(k)}$ ($2\le k< |V|$) denote the set of all $k$-node CISes in $G$.
Two different $k$-node CISes $s_1$ and $s_2$ in $C^{(k)}$ are connected
if and only if they have exactly $k\!-\!1$ nodes in common.
Formally, the undirected graph $G^{(k)}=(C^{(k)},R^{(k)})$ represents the {\em CIS relationships} between all
$k$-node CISes in $G$, where $C^{(k)}$ and $R^{(k)}$ are the node and edge sets
for graph $G^{(k)}$ respectively.
When two $k$-node CISes $s_i$ and $s_j$ in $C^{(k)}$ differ in one and only one node, there exists an edge $(s_i, s_j)$ in graph $G^{(k)}$.
We say that two $k$-node CISes $s_i$ and $s_j$ are reachable if and only if there is at least one path between them in graph $G^{(k)}$, and $G^{(k)}$ is connected if and only if every pair of subgraphs in $C^{(k)}$ is reachable.
Fig.~\ref{fig:exampleCISgraph} shows an example of an original unlabeled graph $G$ and its associated CIS graphs $G^{(2)}$, $G^{(3)}$, and $G^{(4)}$.
Then we have the following theorems.
\begin{theorem}\label{theorem:connected}
If the graph $G$ is connected, then the $k$-node CIS  graph $G^{(k)}$ is connected, $2\le k<|V|$.
\hfill $\square$
\end{theorem}

\begin{theorem}\label{theorem:bipartitedegree}
If the graph $G$ is connected and either non-bipartite, or contains a node with degree larger than two, all $k$-node CIS graphs $G^{(k)}$ are non-bipartite, where $2\le k<|V|$.
\hfill $\square$
\end{theorem}
The proofs of all Theorems in this section are included in the Appendix for completeness.

\noindent
{\bf Remark:}
Theorems~\ref{theorem:bipartitedegree} states
that $G^{(k)}$ is non-bipartite for most connected graph $G$.
Connectedness is critical for removing biases from RW sampling of $G^{(k)}$.
Biases introduced through sampling a bipartite graph using a lazy RW are easily removed.
Biases introduced though sampling via a classical RW can only be removed if the graph is non-bipartite.

\subsection{\textbf{Subgraph random walk (SRW)}}
In this subsection, we propose a sampling method, {\em subgraph random walk} (SRW),
and apply it over graph $G^{(k)}$, $2\le k<|V|$
to estimate concentrations of subgraph classes.
SRW can be viewed as a regular RW over graph $G^{(k)}$.
We define $X(s)\subset C^{(k)}$ as the set of neighbors of the current $k$-node CIS $s$ sampled in $G^{(k)}$.
From an initial $k$-node CIS, SRW randomly selects a CIS from $X(s)$ as the next-hop CIS.
SRW moves to this neighbor CIS and repeats the process.
Clearly $X(s)$ cannot be directly obtained since $G^{(k)}$ is not available in advance.
In what follows we first discuss the number of queries required to compute $X(s)$ at each step.
We observe that $X(s)$ can be obtained using at most $k$ queries, which makes SRW both practical and efficient.
We then show in detail how to apply SRW to sample $k$-node CISes from the original graph $G$ in detail.
Finally we present our method for estimating concentrations of subgraph classes based on sampled CISes.

First, we present a critical observation for analyzing the performance of the SRW:
\emph{Querying nodes in a $k$-node CIS $s$ is suffices to obtain $X(s)$, i.e., the neighbors of $s$ in $G^{(k)}$.}
Denote by $V(s)$ the set of
nodes in $s$ and $E(s)$ the set of edges in $s$.
Denote by $N(s)$ the set of nodes in $V\backslash V(s)$ connected to nodes in $V(s)$.
Let $E^{(N)}(s)$ denote the set of edges between nodes in $N(s)$ and nodes in $V(s)$.
For example, when $s$ is the 3-node CIS consisting of nodes $b$, $c$, $d$ shown in
Fig.~\ref{fig:exampleCISgraph} (c), we have
$d^{(3)}(s)=3$, $V(s)=\{b,c,d\}$, $E(s)=\{(b,c), (b,d), (c,d)\}$, $N(s)=\{a\}$,
$E^{(N)}(s)=\{(a,b), (a,c), (a,d)\}$,
and
$X(s)$ includes three CISes:  the CIS consisting of nodes
1) $a$, $b$, and $c$;
2) $a$, $b$, and $d$; as well as
3) $a$, $c$, and $d$.
Clearly a neighbor of $s$ in $G^{(k)}$ corresponds to a subgraph that
includes $k-1$ nodes in $s$ and one node in $N(s)$.
For each $(k-1)$-node set $\{v_1,\ldots,v_{k-1}\}\subset V(s)$ and each node $u\in N(s)$,
the induced subgraph $s'=(V(s'), E(s'))$ of these $k$ nodes,
i.e., $V(s')=\{v_1, \ldots, v_{k-1}, u\}$ and $E(s')=\{(u,v): u,v\in V(s') \text{ and } (u,v)\in E\}$,
is a neighbor of $s$ in $G^{(k)}$ when $s'$ is a connected graph.
Note that we can obtain $E(s')$ without querying any node in $N(s)$,
since $E(s')=\{(u,v): u,v\in V(s') \text{ and } (u,v)\in E\}=\{(u,v): u,v\in V(s') \text{ and } (u,v)\in E(s)\cup E^{(N)}(s)\}$.
Therefore $X(s)$ can be computed based on $V(s)$, $N(s)$, $E(s)$, $E^{(N)}(s)$, which are all obtained by querying nodes in $s$.

\begin{figure}[htb]
\begin{center}
\includegraphics[width=0.8\textwidth]{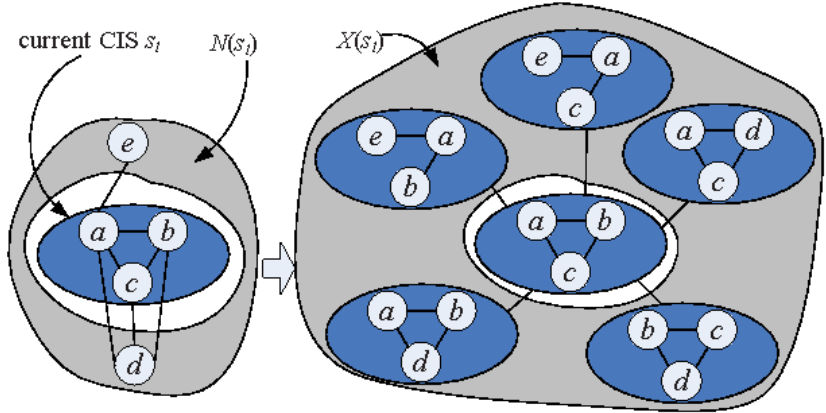}
\caption{An example of applying a SRW to graph $G^{(3)}$.}\label{fig:RWexample}
\end{center}
\end{figure}

\begin{algorithm}
\SetKwFunction{randomNeighbor}{randomNeighbor}
\SetKwFunction{generateGraph}{generateGraph}
\SetKwFunction{connectedGraph}{connectedGraph}
\SetKwInOut{Input}{input}
\SetKwInOut{Output}{output}
\tcc{\footnotesize
$s\in C^{(k)}$.
$N(s)$ is the set of nodes in $V\backslash V(s)$ connected to nodes in $V(s)$.
$E^{(N)}(s)$ is the set of edges between nodes in $N(s)$ and nodes in $V(s)$.}
\Input{$k$-node CIS $s=(V(s), E(s))$, $N(s)$, $E^{(N)}(s)$}
\tcc{\footnotesize $X(s)\subset C^{(k)}$ is the set of neighbors of $s$ in $G^{(k)}$.}
\Output{$X(s)$}
\BlankLine
$X(s)=\{\}$\;
\ForEach {$\{v_1, \ldots, v_{k-1}\}$ $\subset V(s)$}
{
    \ForEach {$u\in N(s)$}
    {
        \tcc{\footnotesize $\generateGraph(\{v_1, \ldots, v_{k-1}\}, u, E(s), E^{(N)}(s))$ returns a graph $s'=(V(s'), E(s'))$,
        whose node set $V(s')=\{v_1, \ldots, v_{k-1}, u\}$,
        and edge set $E(s')=\{(u,v): u,v\in V(s') \text{ and } (u,v)\in E(s)\cup E^{(N)}(s)\}$}
        $s'=\generateGraph(v_1, \ldots, v_{k-1}, u, E(s), E^{(N)}(s))$\;
        \tcc{\footnotesize $\connectedGraph(s')$ returns "True" when $s'$ is a connected graph, and "False" otherwise.}
        \If {$\connectedGraph(s')$}
        {$X(s)=X(s)\cup \{s'\}$}
    }
}
\caption{The pseudo-code of computing $X(s)$. \label{alg:X(s)}}
\end{algorithm}

The SRW algorithm proceeds in steps.
Consider the $i$-th step, $i \ge 1$, and the current node is $s_i \! \in \! C^{(k)}$.
SRW first computes $X(s_i)$, and then selects a CIS randomly from $X(s_i)$ as the next
CIS to visit.
For example, as shown in Fig.~\ref{fig:RWexample}, suppose that the current sampled CIS $s_i$ consists of nodes $a$, $b$, and $c$.
After querying $a$, $b$, and $c$, we obtain $N(s_i)=\{e, d\}$ and the edges between nodes in $N(s_i)$ and nodes in $s_i$.
Then our SRW algorithm computes $X(s_i)$ (i.e., the five CISes connecting to $s_i$ shown in Fig.~\ref{fig:RWexample}),
and randomly select a new CIS from $X(s_i)$ as the next CIS $s_{i+1}$ to sample.
The pseudo-code of computing $X(s_i)$ is shown in Algorithm~\ref{alg:X(s)}.
As mentioned earlier,
$X(s_i)$ is computed without querying any node in $N(s_i)$.
Moreover, since $V(s_{i+1})$ differs from $V(s_i)$ in one and only one node,
SRW only needs to \textbf{\emph{query one node}} in the original graph
$G$ at each step.
Let $d^{(k)}(s)$ be the degree of a $k$-node CIS $s$ in graph $G^{(k)}$, that is the number
of $k$-node CISes connected to $s$.
Formally, SRW then can be modeled as a Markov chain with transition matrix
$\bP^{(k)}=[P^{(k)}_{x,y}]$, $x,y\in C^{(k)}$, where $P^{(k)}_{x,y}$ is defined
as the probability that CIS $y$ is selected as the next sampled
$k$-node CIS given that the current $k$-node CIS is $x$. $P^{(k)}_{x,y}$ is computed as
\begin{equation*}
P^{(k)}_{x,y}=\left\{
\begin{array}{ll}
  \frac{1}{d^{(k)}(x)}, & x\in C^{(k)}, y\in X(x), \\
  0, & \text{otherwise.}
\end{array}
\right.
\end{equation*}
The stationary distribution
$\bpi^{(k)}=\left(\pi^{(k)}(s): s\in C^{(k)}\right)$ of this Markov chain is
\begin{equation*}
\pi^{(k)}(s)=\frac{d^{(k)}(s)}{\sum_{t\in C^{(k)}} d^{(k)} (t)}.
\end{equation*}
SRW is a regular RW over the undirected graph $G^{(k)}$, and
we have the following theorem from~\cite{Lovasz1993,Ribeiro2010}.

\begin{theorem} \label{theorem:srwstationary}
If graph $G^{(k)}$ ($2\le k<|V|$) is non-bipartite and connected, the stationary distribution
for the SRW to be  at a $k$-node CIS $s\in C^{(k)}$ converges to
$\bpi^{(k)}=\left(\pi^{(k)}(s): s\in C^{(k)}\right)$. The probabilities of
a SRW sampling edges in $E^{(k)}$ are equal when the SRW reaches the steady state.
\hfill $\square$
\end{theorem}

\noindent
{\bf Remark:}
As mentioned earlier, in most practical applications the connected non-bipartite assumption over $G^{(k)}$ only implies that the original graph $G$ must have at least one node with degree three or larger and be connected.


Let $C(s)$ denote the subgraph class of a CIS $s$.
$C(s)$ can be easily obtained by using the NAUTY algorithm~\cite{McKay1981,McKay2009}.
Define $\mathbf{1}(\mathcal{P})$ as the indicator function that equals
one when the predicate $\mathcal{P}$ is true,
and zero otherwise.
Let $s_j$, $j>0$, be the $k$-node CIS sampled by the SRW at step $j$.
Using the CISes visited by a SRW after $B > 1$ steps, we use the Horvitz-Thompson estimator~\cite{Ribeiro2010}
to estimate the concentration of subgraph class $C_i^{(k)}$ as:
\begin{equation}
\hat\omega_i^{(k)}
   =\frac{1}{L}\sum_{j=1}^B \frac{\mathbf{1}(C(s_j)= C_i^{(k)})}{d^{(k)}(s_j)},
       \quad \quad 1\le i\le T_k,
\label{eq:estimatorofomegaVSS}
\end{equation}
where $L=\sum_{j=1}^B \frac{1}{d^{(k)}(s_j)}$.

\begin{theorem}\label{theorem:nodeestimatorunbiased}
If $G^{(k)}$ ($2\le k<|V|$) is non-bipartite and connected,
then $\hat\omega_i^{(k)}$ ($1\le i\le T_k$) in Eq.~(\ref{eq:estimatorofomegaVSS})
is an asymptotically unbiased estimator of $\omega_i^{(k)}$.
\hfill $\square$
\end{theorem}

\noindent
{\bf Remark:}
Theorem \ref{theorem:nodeestimatorunbiased}
provides the theoretical basis for producing unbiased estimates
of the concentration of each CIS class in the graph under study.
Proof is found in the appendix.

\subsection{\textbf{Pairwise subgraph random walk (PSRW)}} \label{sec:PSRW}
In this subsection, we use SRW as a building block for our next subgraph statistics method.
Instead of sampling over graph $G^{(k)}$,
we apply SRW to $G^{(k-1)}$, and observe that two adjacent sampled $(k-1)$-node CISes (i.e., an edge in $G^{(k-1)}$) contains exactly $k$ distinct nodes.
Thus, we show how to obtain $k$-node CISes by sampling edges from $G^{(k-1)}$.
Using this property we propose a better sampling method than SRW, which we denote {\em pairwise subgraph random walk} (PSRW).
PSRW samples $k$-node CISes by applying SRW to $G^{(k-1)}$, $2< k\le |V|$.
In what follows we show that PSRW produces more accurate estimates than SRW as observed from
experimental results in Section~\ref{sec:results}.
It remains an open theoretical problem why PSRW significantly  outperforms SRW.
Our conjecture is that it is due to the fact that a RW on $G^{(k-1)}$ converges to its stationary behavior more quickly than {$G^{(k)}$}.
Let $s_j$ ($1\le j\le B$) be the $j$-th $(k-1)$-node CIS sampled by applying a SRW to $G^{(k-1)}$. Consider the edge $(s_j,s_{j+1})$ in $G^{(k-1)}$.
This edge is associated with a $k$-node CIS consisting all nodes
contained in $s_j$ and $s_{j+1}$.
Therefore we obtain $k$-node CISes $s^*_j$ ($1\le j< B$),
where $s^*_j$ is the $k$-node CIS generated by $(s_j,s_{j+1})$.
Fig.~\ref{fig:PSRWexample} shows an example of applying a PSRW to sample 3-node CISes $s_1^*$, $s_2^*$, \ldots from $G$.
We can see that $s_1^*$, $s_2^*$, \ldots are generated based on 2-node CISes $s_1$, $s_2$ sampled by applying a SRW to $G^{(2)}$,
where $G$ and $G^{(2)}$ are shown in Fig.~\ref{fig:exampleCISgraph}.
For any $k$-node CIS $x\in C^{(k)}$, let $I^{(k-1)}(x)$ denote the number of $(k-1)$-node
CISes contained by $x$.
For example, 3-node CIS $s_1^*$ in Fig.~\ref{fig:exampleCISgraph} contains two 2-node CISes:
1) the CIS consisting of nodes $A$ and $B$,
and 2) the CIS consisting of nodes $A$ and $E$.
Thus, $I^{(2)}(s_1^*)=2$.
Similarly we have $I^{(2)}(s_3^*)=3$.
It is easy to show that a $k$-node CIS $x$ associates with $\frac{\left(I^{(k-1)}(x)\right)\left(I^{(k-1)}(x)-1\right)}{2}$ edges in graph $G^{(k-1)}$
and $x$ is sampled by PSRW if and only if at least one of its associated edges in $G^{(k-1)}$ is sampled.
For example, as shown in Fig.~\ref{fig:exampleS2S3} where $k=3$,
$s_3^*$ is sampled when PSRW samples at least one of its associated edges in $G^{(2)}$, i.e., the $\frac{\left(I^{(2)}(s_3^*)\right)\left(I^{(2)}(s_3^*)-1\right)}{2}=3$ red edges in Fig.~\ref{fig:exampleS2S3}.
From Theorem~\ref{theorem:srwstationary}, we know that SRW samples each edge
in $G^{(k-1)}$ with equal probability at steady state,
therefore $k$-node CIS $x$ is sampled with the following probability
\begin{equation*}
\pi_E^{(k)}(x)=\frac{I^{(k-1)}(x)\left(I^{(k-1)}(x)-1\right)}{\sum_{y\in C^{(k)}} I^{(k-1)}(y)\left(I^{(k-1)}(y)-1\right)}.
\end{equation*}
Thus, using the Horvitz-Thompson estimator, we estimate the concentration of subgraph class $C_i^{(k)}$ as follows,
\begin{equation}
\tilde\omega_i^{(k)}\!\!=\!\!
\frac{1}{H} \!\! \sum_{j=1}^{B-1} \!\!
\frac{\mathbf{1}(C(s^*_j)= C_i^{(k)})}{I^{(k-1)}(s^*_j)\left(I^{(k-1)}(s^*_j)-1\right)},
   \quad 1\le i\le T_{k},
\label{eq:estimatorofomegaedge}
\end{equation}
where $H=\sum_{j=1}^{B-1} \left[I^{(k-1)}(s^*_j)\left(I^{(k-1)}(s^*_j)-1\right)\right]^{-1}$.

\begin{figure}[htb]
\begin{center}
\includegraphics[width=0.85\textwidth]{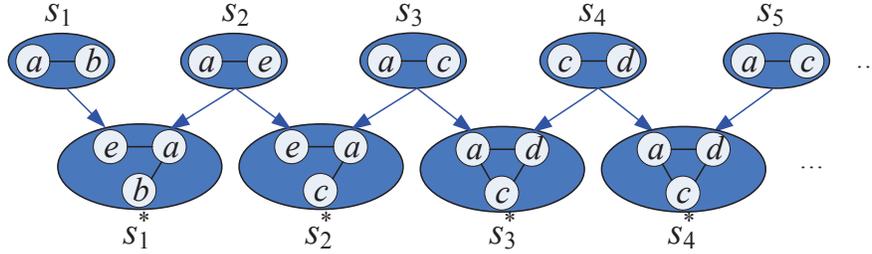}
\caption{An example of applying a PSRW to sample 3-node CISes from $G$.
$s_1$, $s_2$, \ldots, are 2-node CISes sampled by applying a SRW to $G^{(2)}$.
$s_1^*$, $s_2^*$, \ldots, are 3-node CISes generated based on $s_1$, $s_2$, \ldots.
$G$ and $G^{(2)}$ are graphs shown in Fig.~\ref{fig:exampleCISgraph}.
}\label{fig:PSRWexample}
\end{center}
\end{figure}

\begin{figure}[htb]
\begin{center}
\includegraphics[width=0.7\textwidth]{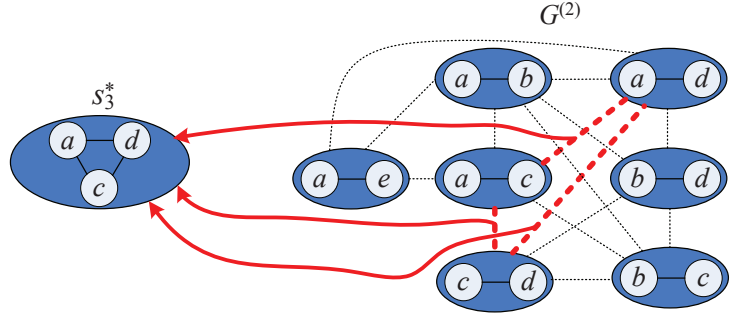}
\caption{3-node CIS $s_3^*$ can be generated by each of the three red edges in $G^{(2)}$.
}\label{fig:exampleS2S3}
\end{center}
\end{figure}

\begin{theorem}\label{theorem:edgeestimatorunbiased}
If $G^{(k)}$ ($2\!\le\! k\!\le\! |V|$) is non-bipartite and connected,
then $\tilde\omega_i^{(k)}$ ($1\!\le\! i\!\le\! T_{k}$)
in Eq.~(\ref{eq:estimatorofomegaedge})
is an asymptotically unbiased estimator of $\omega_i^{(k)}$.
\hfill $\square$
\end{theorem}

\noindent
{\bf Remark:}
We find that PSRW cannot be easily further extended, that is, we cannot sample $k$-node CISes based on applying a SRW to graph $G^{(k')}$, where $k'<k-1$.
This is because it is difficult to analyze and remove sampling errors when
we generate a $k$-node CIS using $k-k'+1$ CISes consequently sampled by a SRW over $G^{(k')}$.
Note that $k-k'+1$ consequently sampled $k'$-node CISes might contain less than $k$ different nodes.

\subsection{Mix Subgraph Sampling (MSS)}
The previous two subsections focus on subgraph classes of a specific size $k$.
Motivated by~\cite{Bhuiyan2012}, we study the problem of estimating the concentrations of subgraph classes of sizes $k-1$, $k$, and $k+1$ simultaneously.
Clearly we can naively solve this problem by applying three independent PRSWs to calculate the concentrations of subgraph classes of sizes $k-1$, $k$, and $k+1$.
However, this naive approach is inefficient.
Next, we propose a more efficient sampling method, mix subgraph sampling (MSS), which requires fewer queries to achieve the same estimation accuracy.
MSS samples $k$-node CISes by applying a SRW to $G^{(k)}$.

Let $s_j$, $j>0$, be the $k$-node CIS sampled by the SRW at step $j$.
Using the CISes visited by a SRW after $B > 1$ steps, MSS estimates the concentrations of $k$-node subgraph classes using Eq.~(\ref{eq:estimatorofomegaVSS}).
Similar to PSRW, MSS uses $s_j$ ($1\le j\le B$) to generate (k+1)-node CISes, and then estimates the concentrations of $(k+1)$-node subgraph classes using Eq.~(\ref{eq:estimatorofomegaedge}).
Last, we show how MSS estimates the concentrations of $(k-1)$-node subgraph classes based on $s_j$, $1\le j\le B$.
Let $C^{(k-1)}(s)$ denote the set of $(k-1)$-node CISes contained in a CIS $s$.
For example, if $s$ is the 4-node CIS consisting of nodes $a$, $b$, $c$, and $e$ shown in Fig.~\ref{fig:exampleCISgraph},
then $S^{(3)}(s)$ consists of three 3-node CISes:
1) the CIS consisting of nodes $a$, $b$, and $c$;
2) the CIS consisting of nodes $a$, $b$, and $e$;
3) the CIS consisting of nodes $b$, $c$, and $e$.
Define $O^{(k)}(s')$ as the set of $k$-node CISes that contain a CIS $s'$.
For example, if $s'$ is the 3-node CIS consisting of nodes $a$, $c$, and $e$ shown in Fig.~\ref{fig:exampleCISgraph},
then $O^{(4)}(s')$ consists of two 4-node CISes:
1) the CIS consisting of nodes $a$, $b$, $c$, and $e$;
2) the CIS consisting of nodes $a$, $c$, $d$, and $e$.
Finally MSS estimates the concentration of subgraph class $C_i^{(k-1)}$ as follows,
\begin{equation}
\breve{\omega}_i^{(k-1)}=\frac{1}{Q} \sum_{j=1}^B
\frac{1}{d^{(k)}(s_j)} \sum_{s'\in C^{(k-1)}(s_j)} \frac{\mathbf{1}(C^{(k-1)}(s')= C_i^{(k-1)})}{|O^{(k)}(s')|},
   \quad 1\le i\le T_{k-1},
\label{eq:estimatorofomegareduce}
\end{equation}
where $Q=\sum_{j=1}^B
\frac{1}{d^{(k)}(s_j)} \sum_{s'\in C^{(k-1)}(s_j)} \frac{1}{|O^{(k)}(s')|}$.
\begin{theorem}\label{theorem:reduceestimatorunbiased}
If $G^{(k)}$ ($3\!\le\! k\!\le\! |V|$) is non-bipartite and connected,
then $\breve\omega_i^{(k-1)}$ ($1\!\le\! i\!\le\! T_{k-1}$)
in Eq.~(\ref{eq:estimatorofomegareduce})
is an asymptotically unbiased estimator of $\omega_i^{(k-1)}$.
\hfill $\square$
\end{theorem}

\section{\textbf{Data Evaluation}} \label{sec:results}
In this section, we first introduce our experimental datasets and a comparison model, which is used to evaluate the performance of our methods for characterizing CIS classes of a specific size $k$ in comparison with state-of-the-art methods.
Then we present the experimental results of PSRW for $k \in \{3,\ldots,6\}$.
Last, we compare the special case MSS $k=3, 4, 5$  against GUISE~\cite{Bhuiyan2012}.
Our experiments are conducted on a Dell Precision T1650 workstation with an Intel Core i7-3770 CPU 3.40 GHz processor and 8 GB DRAM memory.
\subsection{\textbf{Datasets}}
Our experiments are performed on a variety of publicly available datasets taken from the Stanford Network Analysis Platform (SNAP)\footnote{www.snap.stanford.edu}.
We start by evaluating the performance of our methods in characterizing $3$-node CISes over four million-node graphs: Flickr, Pokec, LiveJournal, and YouTube, contrasting our results with ground truth computed through brute force.
Since it is computationally intensive to calculate the ground-truth of $k$-node CIS classes for $k \! \ge \!4$ in large graphs,
the experiments for $k$-CISes, $k \geq 4$, are performed on three relatively small graphs Epinions, Slashdot, and Gnutella, where computing the ground-truth by brute force is feasible.
We also specifically evaluate the performance of our methods for characterizing signed CIS classes in the Epinions and Slashdot signed graphs.
Flickr, LiverJournal, and YouTube are popular photo, blog, and video sharing websites respectively,
where a user can subscribe to other user updates such as photos, blogs, and videos.
Pokec is the most popular on-line social network in Slovakia, and has been
in existence for more than ten years.
These four networks can be represented by directed graphs,
where nodes representing users and a directed edge from $u$ to $v$
represents that user $u$ subscribes to user $v$ or $u$ tags user $v$ as a friend.
Epinions is a who-trust-whom OSN providing consumer reviews, where a directed edge from $u$ to $v$ represents that user $u$ trusts user $v$.
Slashdot is a technology-related news website for its specific user community, where a directed edge from $u$ to $v$ represents that user $u$ tags user $v$ as
a friend or foe.
Epinions and Slashdot networks can be represented by signed graphs,
where a positive edge from $u$ to $v$ indicates that $u$ trusts $v$ or $u$ tags user $v$ a friend,
and a negative edge from $u$ to $v$ indicates that $u$ distrusts $v$ or $u$ tags user $v$ a foe.
Gnutella is a peer-to-peer file sharing network.
Nodes represent users in the Gnutella network and edges represent connections between the Gnutella users.
In the following experiments, we evaluate our proposed methods on the largest connected component (LCC) of these graphs (summarized in Table~\ref{tab:datasets}).

\begin{table}[htb]
\begin{center}
\tbl{Overview of graph datasets used in our simulations.\label{tab:datasets}}{
\begin{tabular}{|c|ccc|}
\hline
\multirow{2}{*}{Graph} &  & LCC & \\
\cline{2-4}
&nodes&edges&directed-edges\\
\hline
Flickr~\cite{MisloveIMC2007}&1,624,992&15,476,835&22,477,014\\
Pokec~\cite{Takac2012}&1,632,805&22,301,964&30,622,564\\
LiveJournal~\cite{MisloveIMC2007}&5,284,459&48,688,097&76,901,758\\
YouTube~\cite{MisloveIMC2007}&1,134,890&2,987,624&4,942,035\\
Epinions~\cite{Richardson2003}&119,130&704,267&833,390\\
Slashdot~\cite{LeskovecIM2009}&77,350&416,695&516,575\\
Gnutella~\cite{LeskovecIM2009}&6,299&20,776&20,776\\
\hline
\end{tabular}}
\end{center}
\begin{tabnote}
\Note{Note:}{``directed-edges" refers to the number of directed edges in a
directed graph, ``edges" refers to the number of edges in an
undirected graph, and ``LCC" refers to the largest connected component
of a given graph.}
\end{tabnote}
\end{table}

\subsection{\textbf{Comparison model}}
First we introduce the error metric used to compare the different sampling methods.
Mean square error (MSE) is a common measure to quantify the error of an estimate $\hat{\omega}$ with respect to its true value $\omega>0$.
It is defined as
$\text{MSE}(\hat{\omega})=\text{E}[(\hat{\omega}-\omega)^2]=\text{var}(\hat{\omega})+\left(\text{E}[\hat{\omega}]-\omega\right)^2$.
We can see that $\text{MSE}(\hat{\omega})$ decomposes into a sum of the variance and bias of the estimator $\hat{\omega}$, both quantities are important and need to be as small as possible to achieve good estimation performance.
When $\hat{\omega}$ is an unbiased estimator of $\omega$,
then $\text{MSE}(\hat{\omega})= \text{var}(\hat{\omega})$.
In our experiments, we study the normalized root mean square error (NRMSE) to measure the relative error of the estimator $\hat{\omega}_i$ of the subgraph class concentration $\omega_i$, $i=1,2,\dots$.
$\text{NRMSE}(\hat{\omega}_i)$ is defined as:
\[
\text{NRMSE}(\hat{\omega}_i)=\frac{\sqrt{\text{MSE}(\hat{\omega}_i)}}{\omega_i}, \qquad i=1,2,\dots.
\]
When $\hat{\omega}_i$ is an unbiased estimator of $\omega_i$,
then $\text{NRMSE}(\hat{\omega}_i)$ is equivalent to the normalized standard error of $\hat{\omega}_i$, i.e., $\text{NRMSE}(\hat{\omega}_i)= \sqrt{\text{var}(\hat{\omega}_i)}/\omega_i$.
Note that our metric uses
the relative error. Thus, when $\omega_i$ is small, we consider
values as large as $\text{NRMSE}(\hat{\omega}_i)$ to be acceptable.
In all our experiments, we average the estimates and calculate their
NRMSEs over 1,000 runs.

We compute the NRMSE of our methods for estimating concentrations of CIS classes of specific size $k$,
in comparison with that of two state-of-the-art algorithms FANMOD~\cite{Wernicke2006}
and GUISE in~\cite{Bhuiyan2012} under the constraint that the number of queries cannot exceed $B^*$.
As mentioned earlier, issues arise when comparing our methods of estimating concentrations of CIS classes of a specific size $k$
to that of the method in~\cite{Bhuiyan2012},
as the latter wastes most queries to sample CISes of size not equal to $k$.
To address this problem, we use adapt the MHRW of GUISE~\cite{Bhuiyan2012} to focus on subgraphs of size $k$,
which we name metropolis-Hastings subgraph random walk (MHSRW).
Later we compare MSS $k=3,4,5$ directly to GUISE~\cite{Bhuiyan2012} showing that MSS is significantly more accurate than GUISE.

To sample $k$-node CISes, MHSRW works as follows:
At each step, MHSRW randomly selects a $k$-node CIS $y$ from $X(x)$, the set of neighbors of the current $k$-node CIS $x$ on CIS relationship graph $G^{(k)}$,
and accepts the move with probability $\min\left \{1, \frac{d^{(k)}(y)}{d^{(k)}(x)}\right\}$. Otherwise, it remains at $x$.
MHSRW samples $k$-node CISes uniformly when it reaches the steady state.
Based on CIS samples $s_j$ ($1\le j\le B$), MHSRW estimates the concentration of
subgraph class $C_i^{(k)}$ for graph $G_d$ as follows,
\begin{equation*}
\breve{\omega}_i^{(k)}=\frac{1}{B}\sum_{j=1}^B \mathbf{1}(C(s_j)= C_i^{(k)}), \qquad 1\le i\le T_k.
\end{equation*}

\subsection{\textbf{Results of estimating 3-node CIS class concentrations}}
In this subsection we show results of estimating 3-node CIS class concentrations for undirected, directed, and signed graphs respectively.

1) \textbf{3-node undirected CIS classes}: Fig.~\ref{fig:undirectedthreenodes} shows the results of estimating
$\omega_2^{(3)}$, the concentration of the 3-node undirected CIS class 2
(or the triangle as shown in Fig.~\ref{fig:34nodeclasses} (a))
for Flickr, Pokec, LiveJournal, and YouTube graphs,
where $B^*$ is the number of queries, i.e., the number of distinct nodes required to query in the original graph $G$.
The true value of $\omega_2^{(3)}$ for
Flickr, Pokec, LiveJournal, and YouTube are 0.0404, 0.0161, 0.0451, and 0.0021 respectively.
The results show that PSRW exhibits the smallest errors, which are almost
an order of magnitude less than errors of MHSRW and FANMOD for
Flickr and Pokec graphs.
SRW is more accurate than MHSRW and FANMOD
but less accurate than PSRW.
Note that PSRW uses only $B^*\!=\!3\times 10^3$ queries and still
exhibits smaller errors than the other methods
that use one order of magnitude more queries $B^*\!>\!3\times 10^4$.
Hence, PSRW reduces more than 10-fold the number of queries
required to achieve the same estimation accuracy.
Meanwhile we observe that an order of magnitude increase in $B^*$ roughly
decreases the error by $1/\sqrt {10}$ for all methods studied.
Fig.~\ref{fig:convergerate} plots the evolution of $\omega_2^{(3)}$
estimates as a function of $B$ (the number of sampling steps) for one run.
We observe that PSRW {\em converges} to the value of
$\omega_2^{(3)}$ when $10^3$, $10^3$, $2\times 10^5$, and $10^4$ CISes are sampled for Flickr, Pokec, LiveJournal, and YouTube respectively, and is much more quickly than the other methods.
To compare the performances of these methods after the random walks have entered the stationary regime,
Fig.~\ref{fig:cmpstationary} shows the results of estimating
$\omega_2^{(3)}$ based on $10^4$ CISes sampled after $10^6$ steps.
We can see that PSRW outperforms the other methods at steady state.

2) \textbf{3-node directed CIS classes}: Fig.~\ref{fig:groundtruth3d} shows the concentrations of 3-node directed
CIS classes for Flickr, Pokec, LiveJournal, and YouTube graphs, and the subgraph classes and their
associated IDs are listed in Fig.~\ref{fig:3nodeclasses}.
The total numbers of 3-node CISes are $1.4\times 10^{10}$,
$2.0\times 10^9$, $6.9\times 10^9$, and $1.5\times 10^9$ for Flickr, Pokec, LiveJournal, and YouTube respectively.
Fig.~\ref{fig:reducethree} compares concentrations estimates of
3-node directed CIS classes for different methods under the same number of queries
$B^*\!=\!10,000$.  The results show that subgraph classes with smaller
concentrations have larger NRMSEs.  PSRW is significantly more
accurate than the other methods for most subgraph classes.
SRW is not shown in the plots but its performance lies again somewhere between MHSRW and FANMOD.

\begin{figure*}[htb]
\center
\subfigure[Flickr]{
\includegraphics[width=0.49\textwidth]{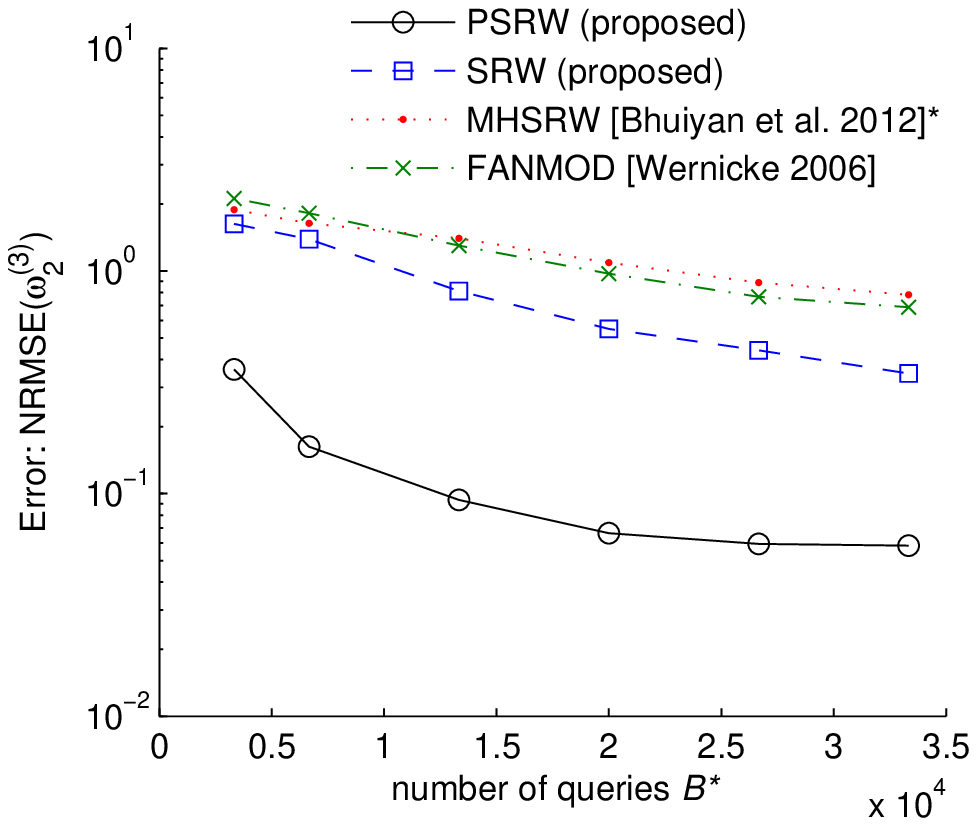}}
\subfigure[Pokec]{
\includegraphics[width=0.49\textwidth]{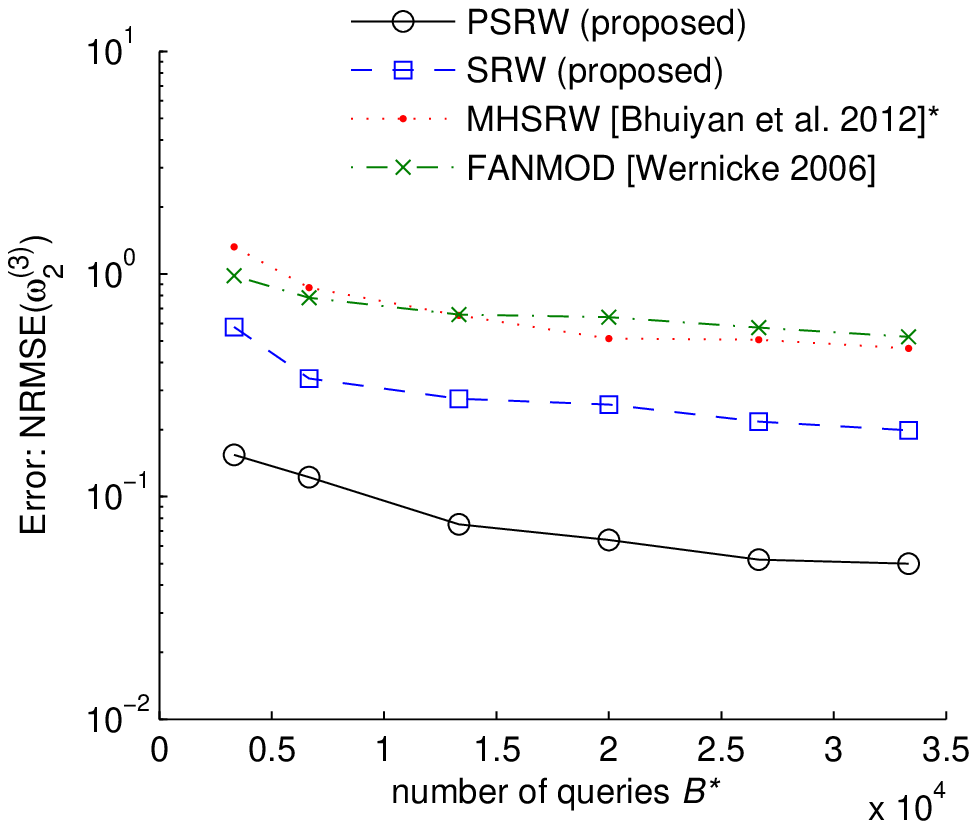}}
\subfigure[LiveJournal]{
\includegraphics[width=0.49\textwidth]{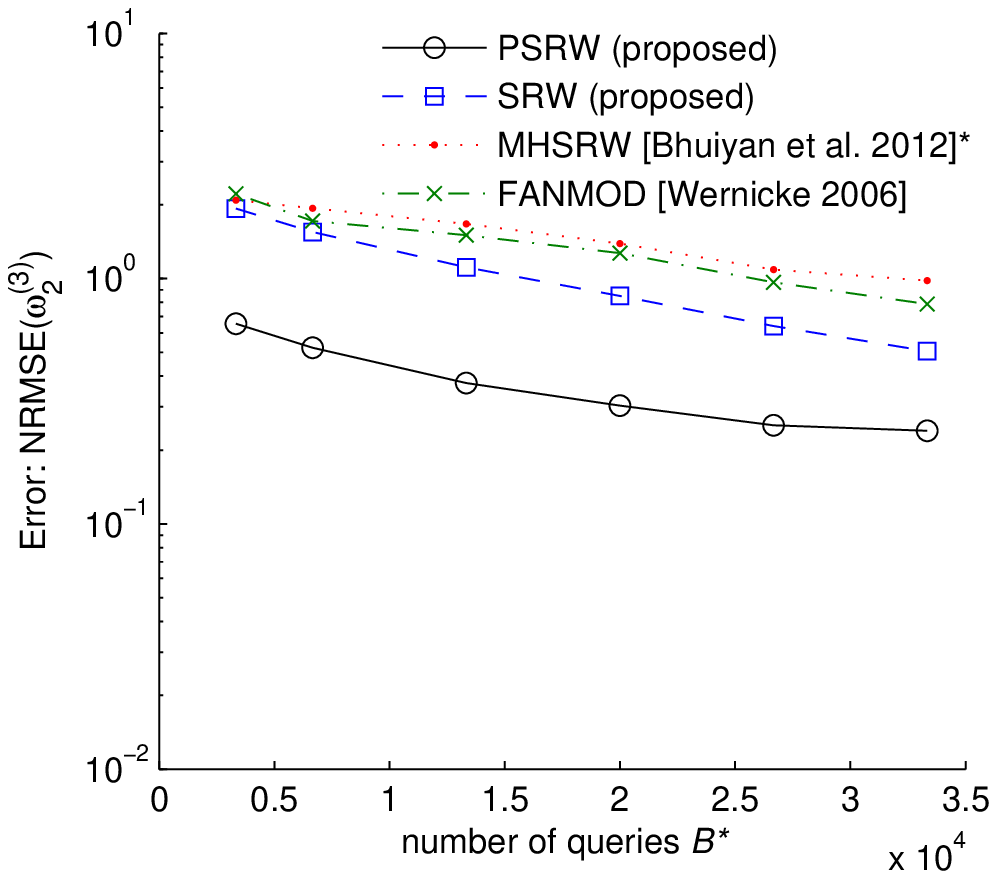}}
\subfigure[YouTube]{
\includegraphics[width=0.49\textwidth]{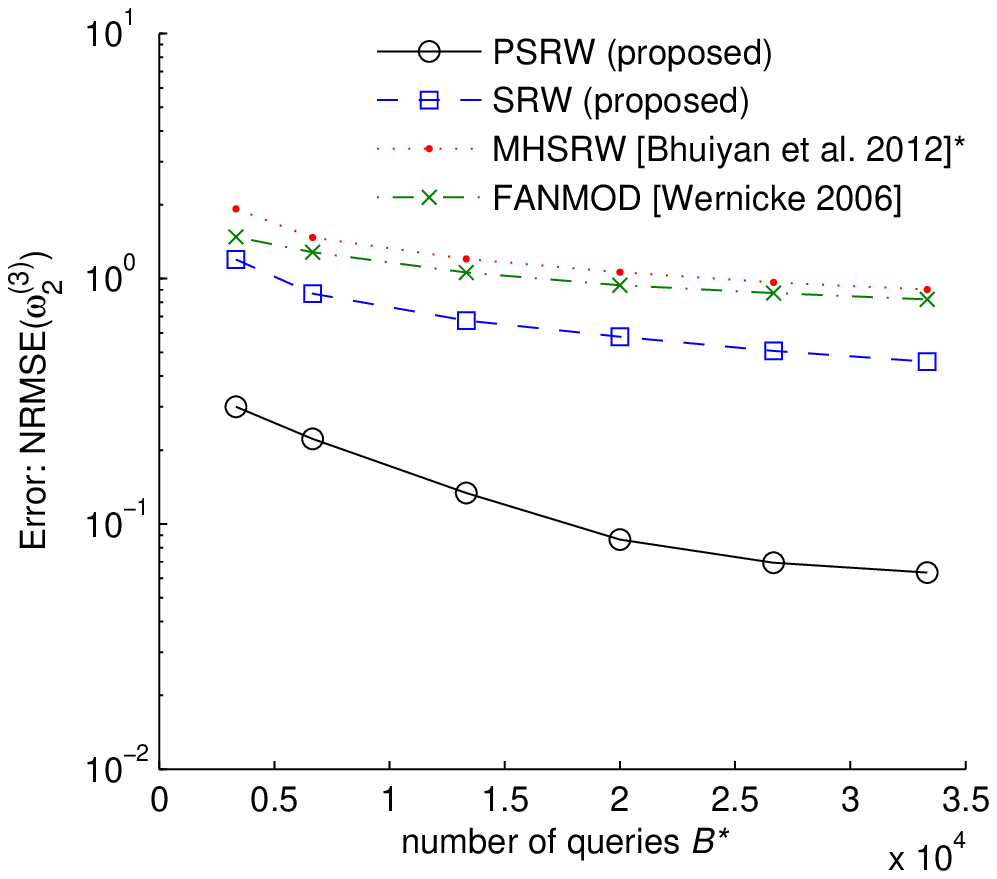}}
\caption{(Flickr, Pokec, LiveJournal, and YouTube) Compared NRMSEs of concentration estimates of 3-node undirected CIS classes for different methods.}\label{fig:undirectedthreenodes}
\end{figure*}

\begin{figure*}[htb]
\center
\subfigure[Flickr, $\omega_2^{(3)}=0.0404$.]{
\includegraphics[width=0.49\textwidth]{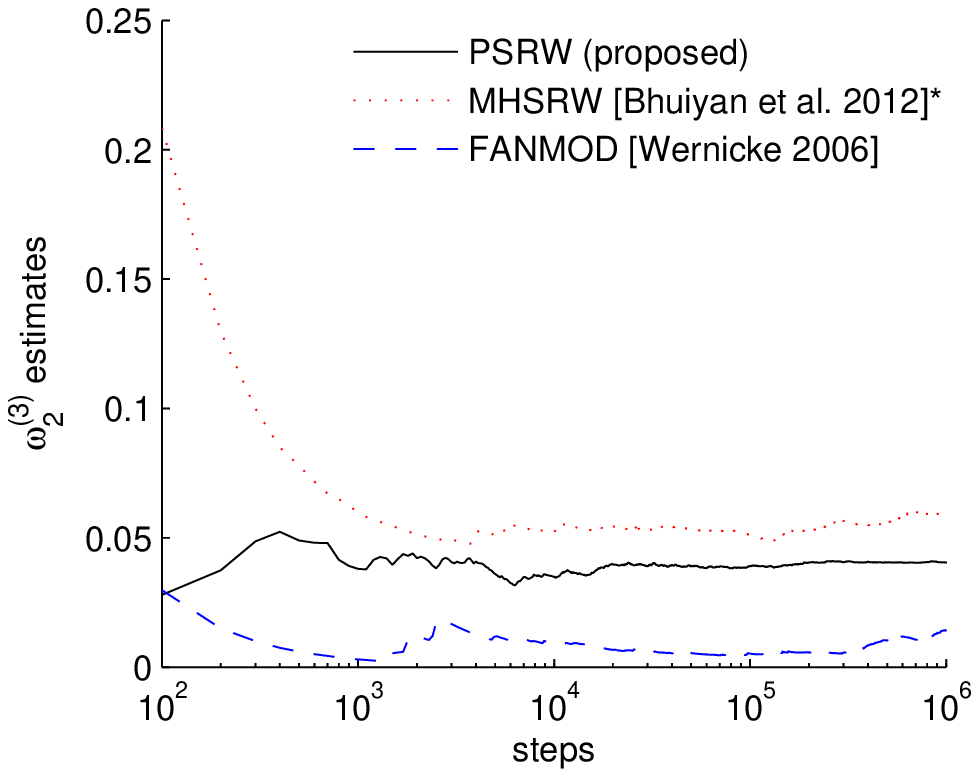}}
\subfigure[Pokec, $\omega_2^{(3)}=0.0161$.]{
\includegraphics[width=0.49\textwidth]{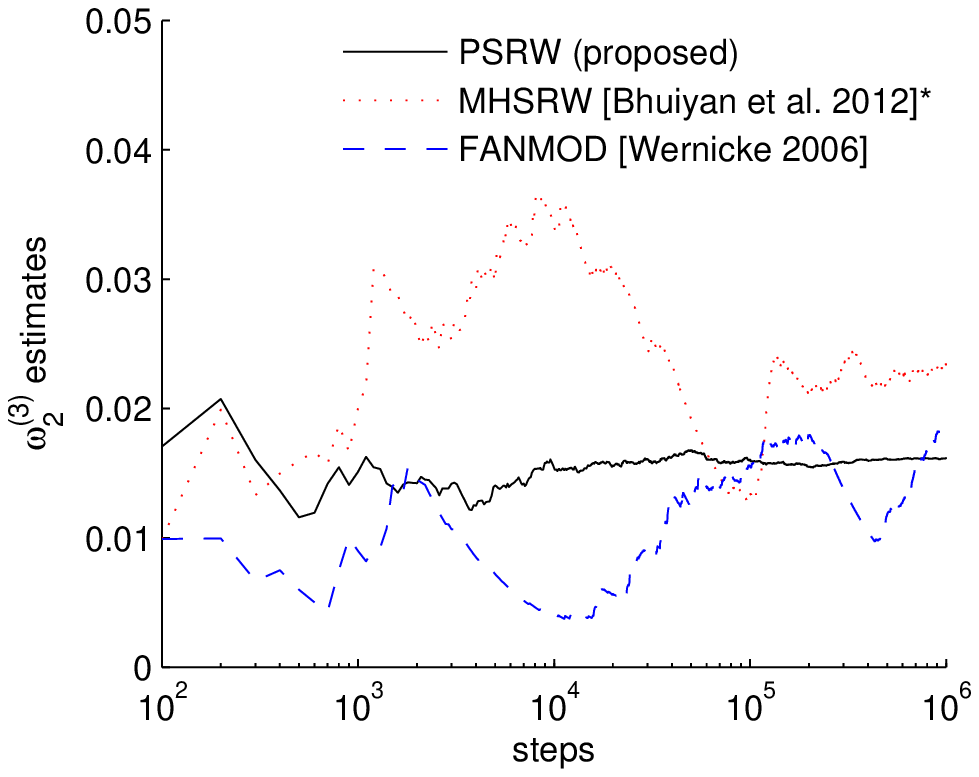}}
\subfigure[LiveJournal, $\omega_2^{(3)}=0.0451$.]{
\includegraphics[width=0.49\textwidth]{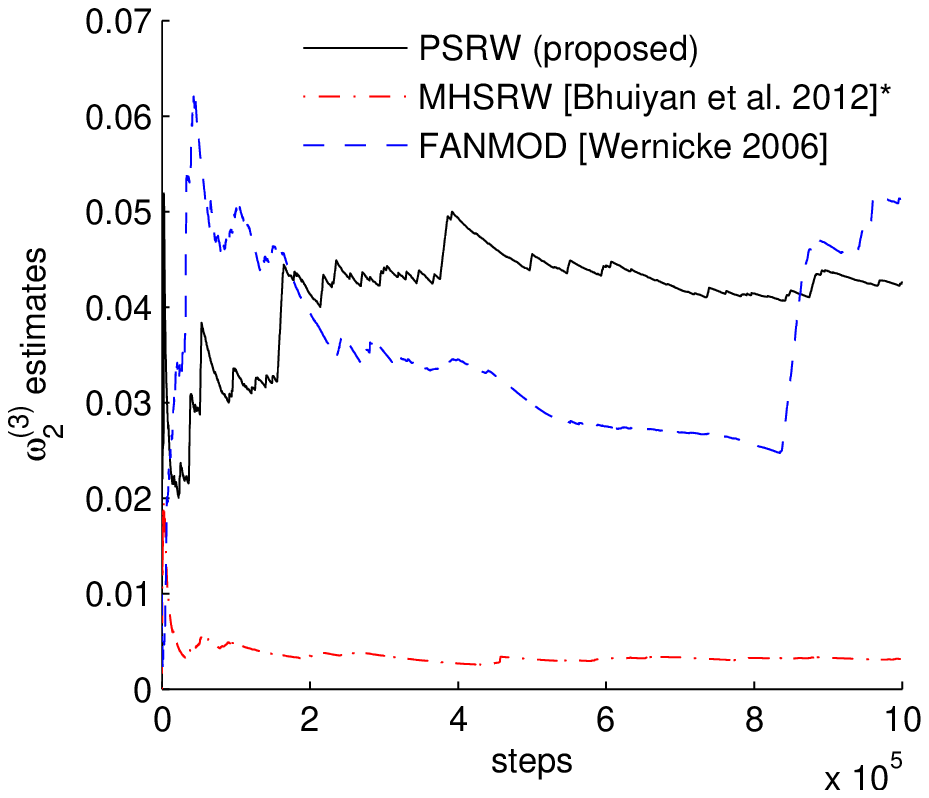}}
\subfigure[YouTube, $\omega_2^{(3)}= 0.0021$.]{
\includegraphics[width=0.49\textwidth]{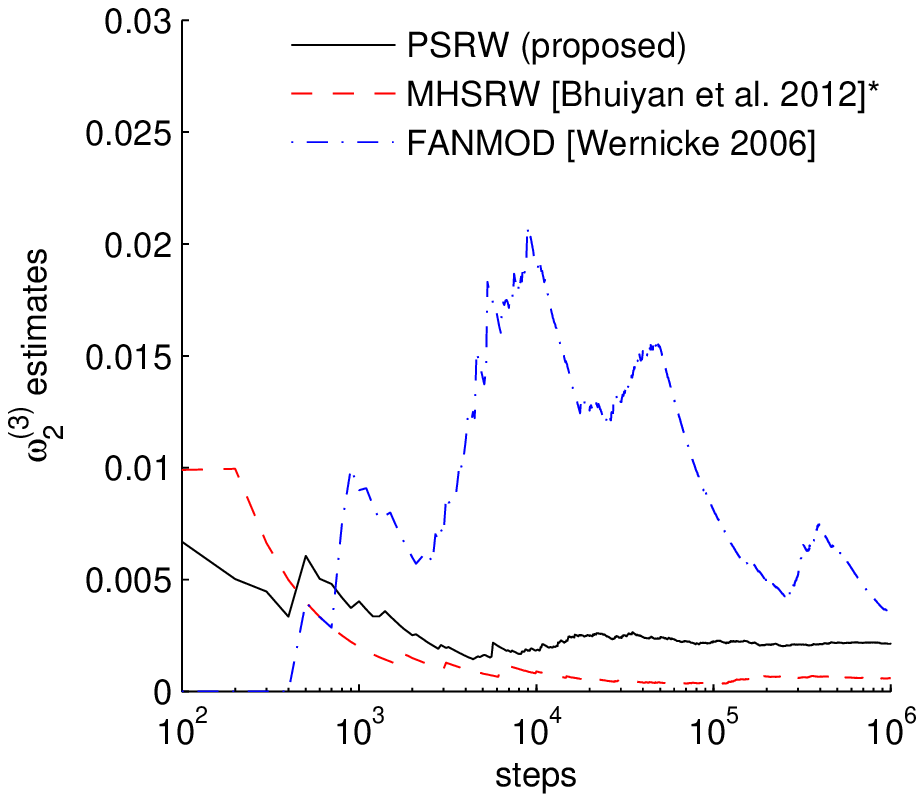}}
\caption{(Flickr, Pokec, LiveJournal, and YouTube) Compared $\omega_2^{(3)}$ estimates of 3-node undirected CIS classes for different methods.}\label{fig:convergerate}
\end{figure*}

\begin{figure}[htb]
\begin{center}
\includegraphics[width=0.7\textwidth]{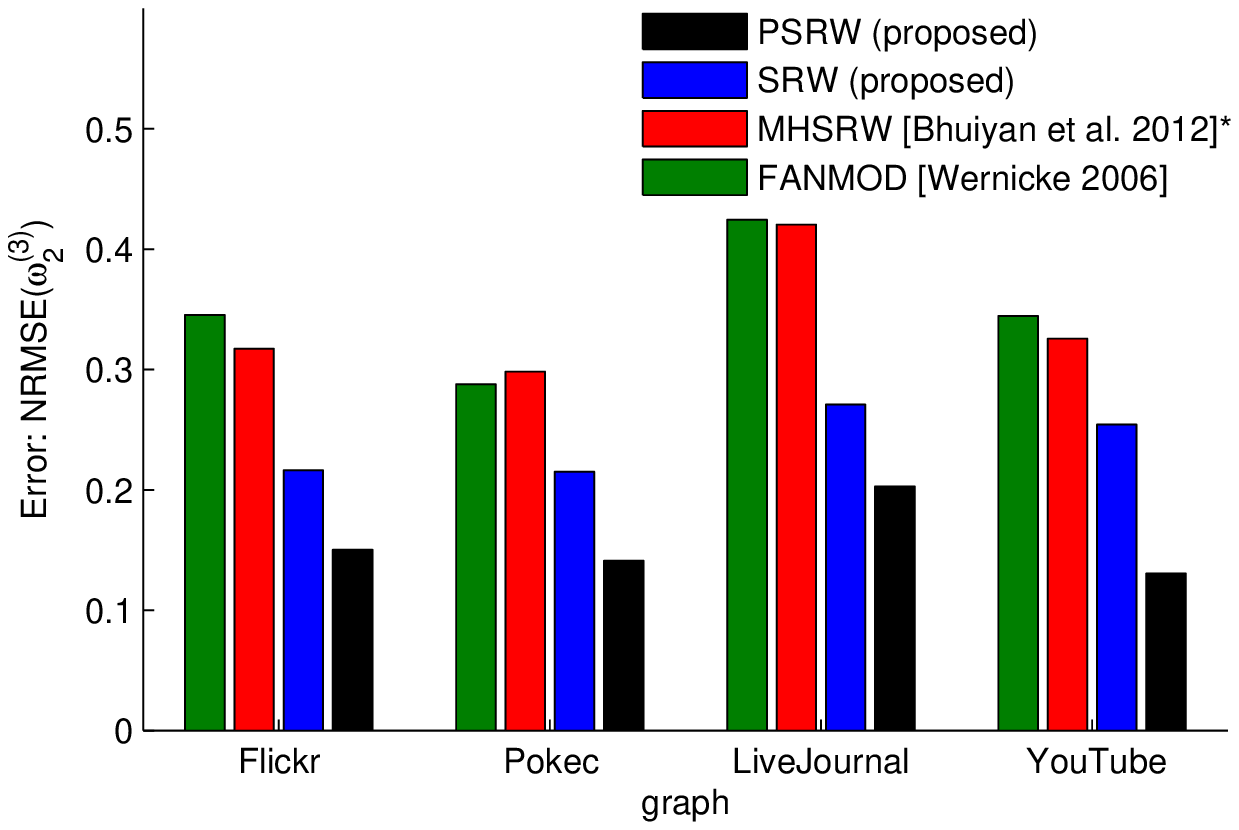}
\caption{(Flickr, Pokec, LiveJournal, and YouTube) Compared NRMSEs of concentration estimates of 3-node undirected CIS classes for different methods at the stationary state.}\label{fig:cmpstationary}
\end{center}
\end{figure}

\begin{figure}[htb]
\begin{center}
\includegraphics[width=0.99\textwidth]{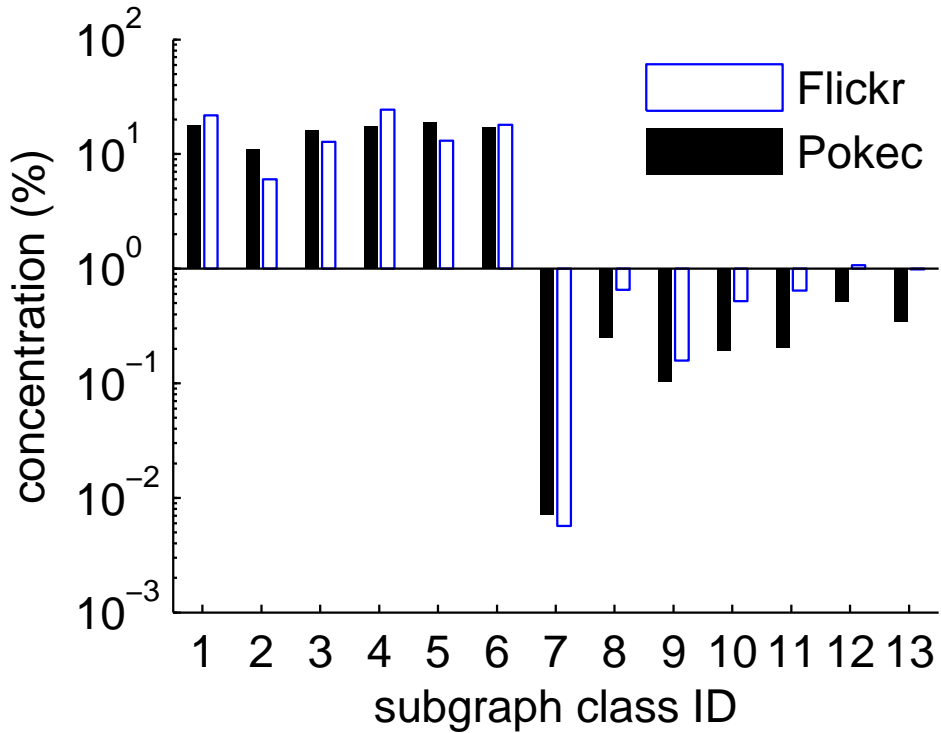}
\caption{(Flickr, Pokec, LiveJournal, and YouTube) Concentrations of 3-node directed CIS classes.}\label{fig:groundtruth3d}
\end{center}
\end{figure}

\begin{figure*}[htb]
\center
\subfigure[Flickr]{
\includegraphics[width=0.49\textwidth]{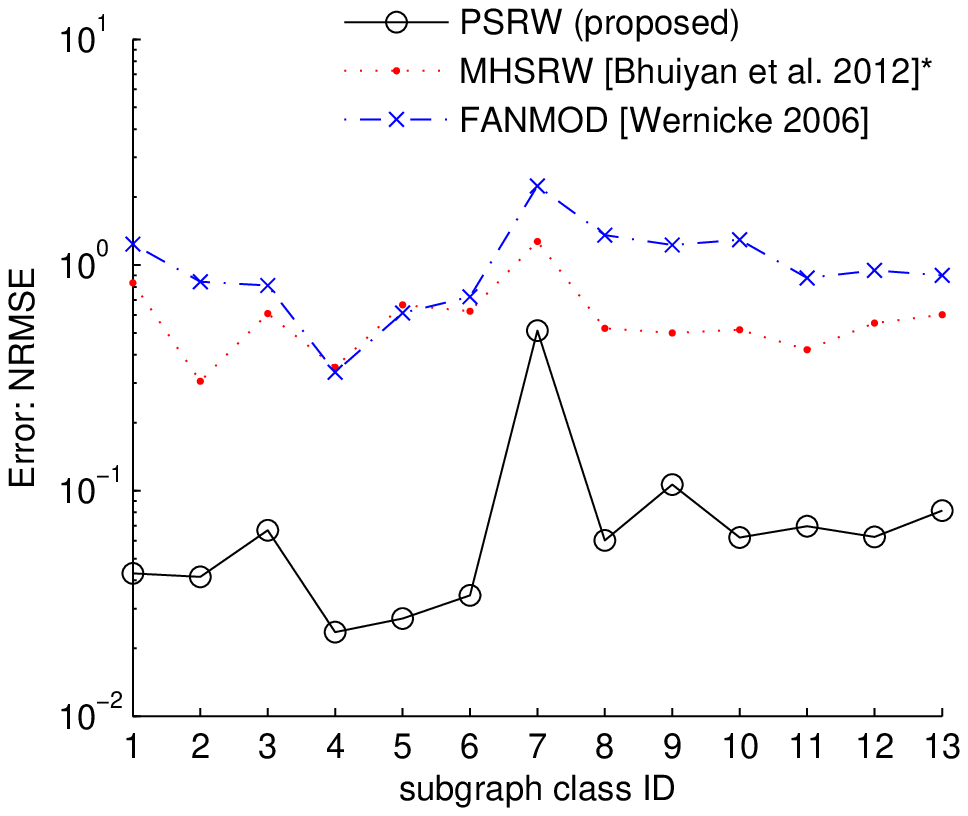}}
\subfigure[Pokec]{
\includegraphics[width=0.49\textwidth]{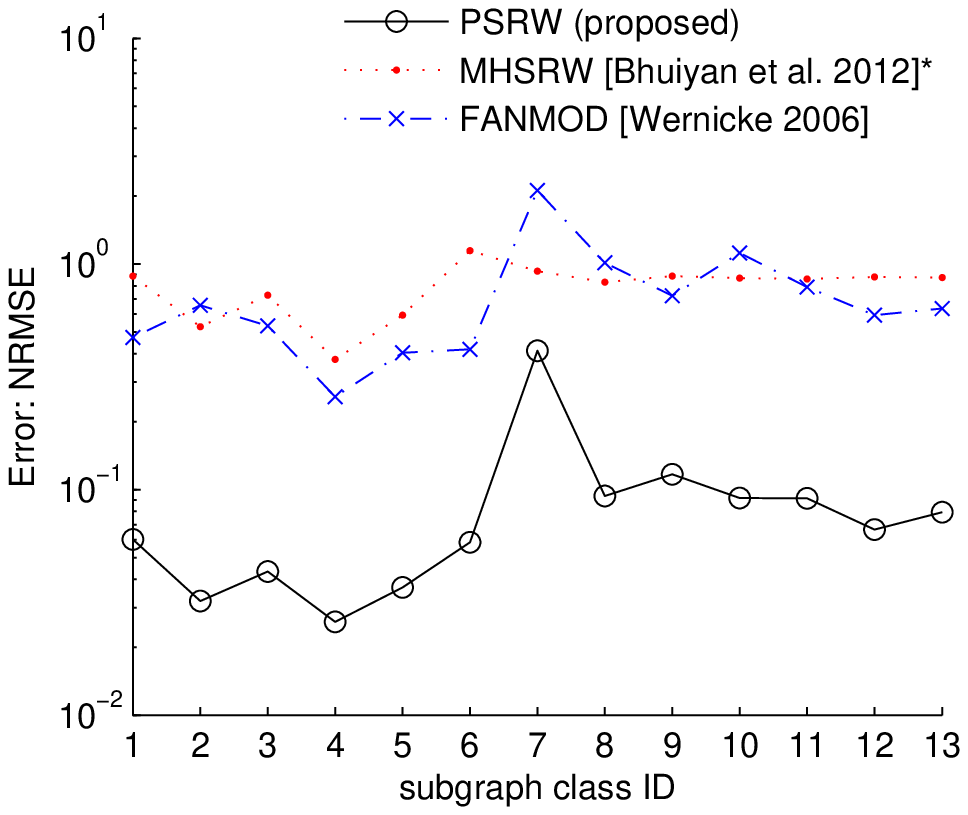}}
\subfigure[LiveJournal]{
\includegraphics[width=0.49\textwidth]{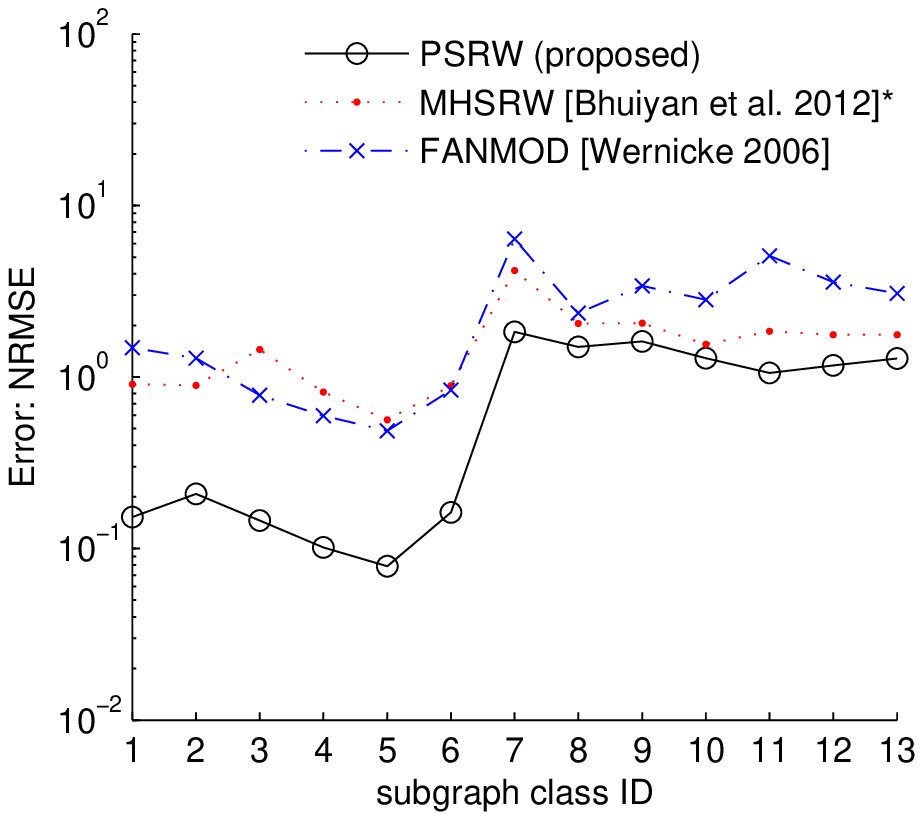}}
\subfigure[YouTube]{
\includegraphics[width=0.49\textwidth]{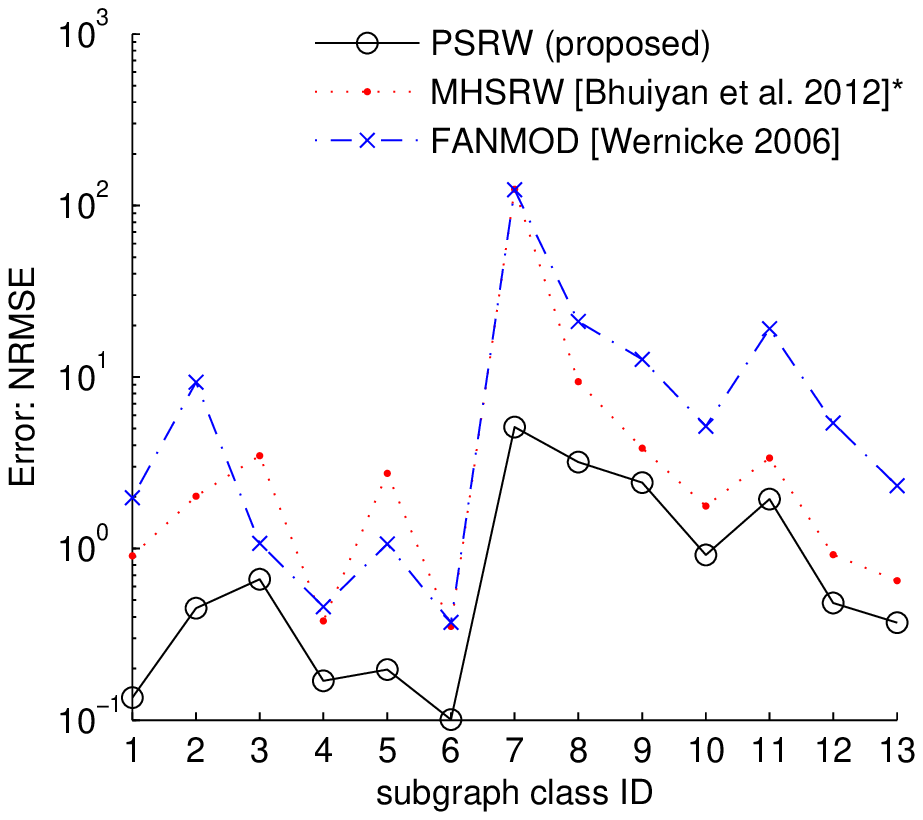}}
\caption{(Flickr, Pokec, LiveJournal, and YouTube) Compared NRMSEs of concentration estimates of 3-node directed CIS classes for different methods under the same number of queries $B^*=10,000$.}\label{fig:reducethree}
\end{figure*}

3) \textbf{3-node signed and undirected CIS classes}: Figure~\ref{fig:groundtruth3sign} shows the concentrations of 3-node signed and
undirected CIS classes (as listed in
Fig.~\ref{fig:signclasses3nodes}) for Epinions and Slashdot graphs.
Epinions and Slashdot graphs have
$1.7\times 10^{8}$ and $6.7 \times 10^{7}$ signed and
undirected 3-node CISes respectively.
Fig.~\ref{fig:signreducethree} shows the estimated concentrations of
signed and undirected 3-node CIS classes for different methods
under $B^*\!=\!2,000$ queries.
The results show that subgraph classes with smaller concentrations have
larger NRMSEs.
All NRMSEs given by PSRW are much smaller than one for all subgraph classes.
PSRW is almost four times more accurate than MHSRW and FANMOD.
MHSRW exhibits slightly smaller errors than FANMOD for most subgraph classes.

\begin{figure}[htb]
\begin{center}
\includegraphics[width=0.45\textwidth]{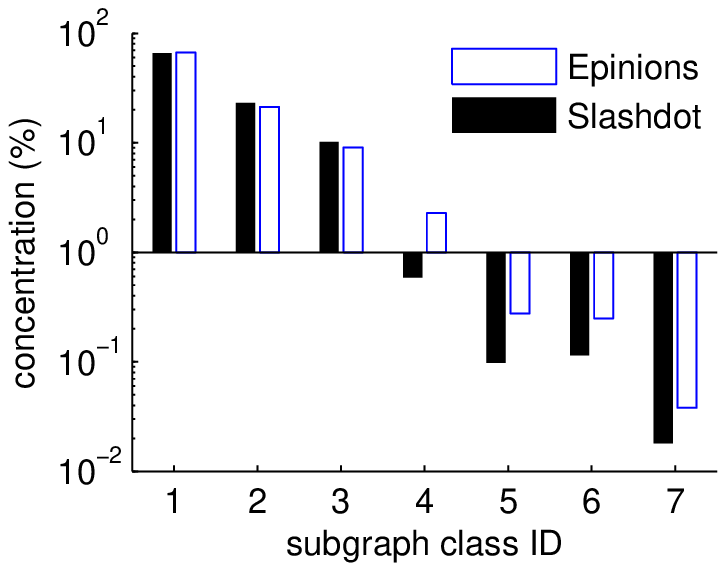}
\caption{(Epinons and Slashdot) Concentrations of 3-node signed and undirected CIS classes.}\label{fig:groundtruth3sign}
\end{center}
\end{figure}

\begin{figure*}[htb]
\center
\subfigure[Epinons]{
\includegraphics[width=0.49\textwidth]{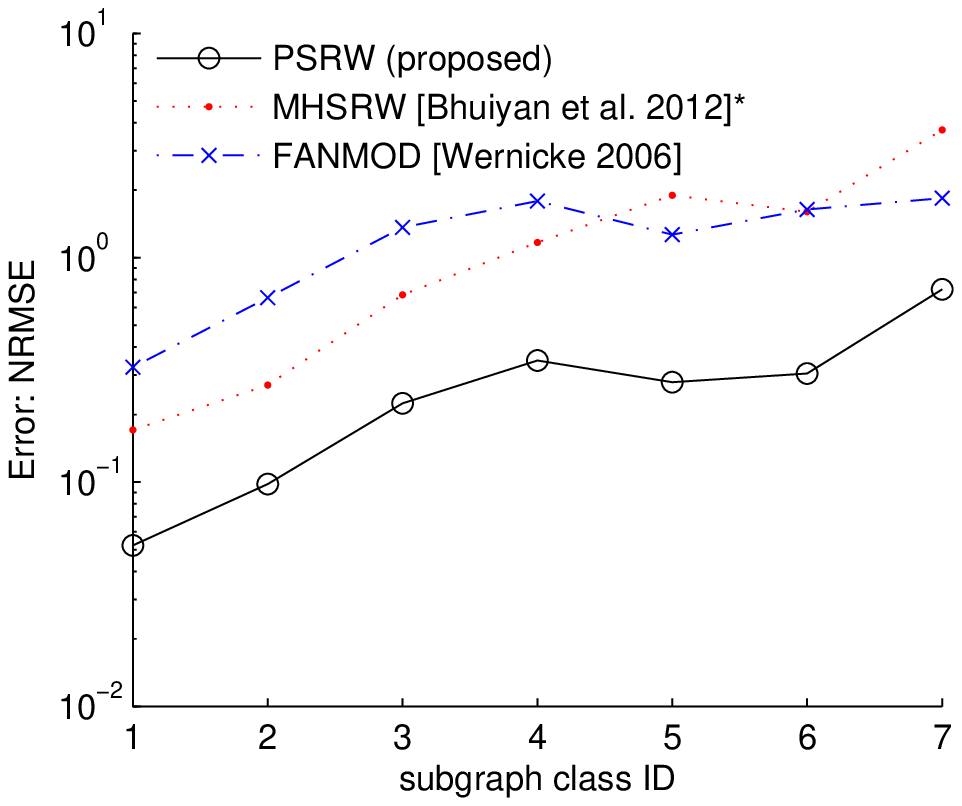}}
\subfigure[Slashdot]{
\includegraphics[width=0.49\textwidth]{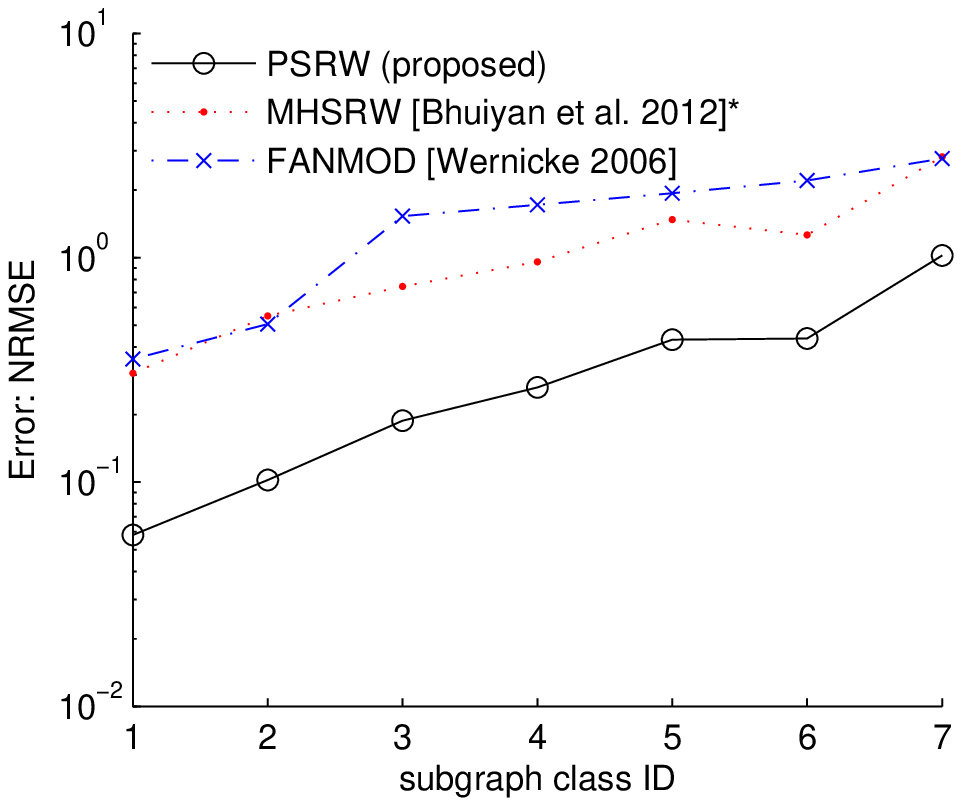}}
\caption{(Epinons and Slashdot) Compared NRMSEs of concentration estimates of 3-node signed and undirected CIS classes for different methods under the same number of queries $B^*=2,000$.}\label{fig:signreducethree}
\end{figure*}

\subsection{\textbf{Results of estimating 4-node CIS class concentrations}}
Figure~\ref{fig:groundtruth4u} shows the concentrations of
4-node undirected CIS classes (as listed in Fig.~\ref{fig:34nodeclasses} (b))
for Epinions, Slashdot, and Gnutella graphs.
Epinions, Slashdot, and Gnutella graphs have $2.5\times 10^{10}$, $2.1\times 10^{10}$, and $1.0\times 10^7$
undirected four-node CISes respectively.
Fig.~\ref{fig:undirectedfourclass} shows the estimated concentrations of
undirected four-node CIS classes for different methods under $B^*\!=\!2,000$ queries.
The results show that all NRMSEs given by PSRW are
smaller than 0.4 for subgraph classes 1 to 5.
PSRW is significantly more accurate than the other methods.
\begin{figure}[htb]
\begin{center}
\includegraphics[width=0.7\textwidth]{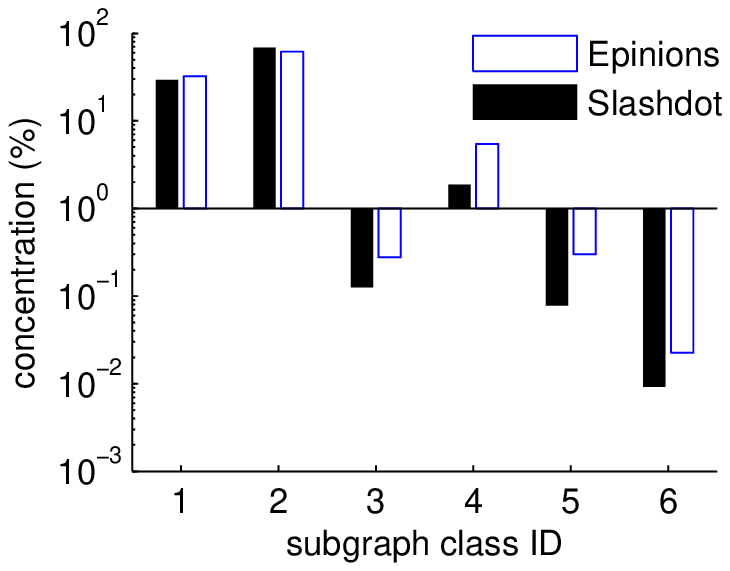}
\caption{(Epinons, Slashdot, and Gnutella) Concentrations of the 4-node undirected CIS classes.}\label{fig:groundtruth4u}
\end{center}
\end{figure}

\subsection{\textbf{Results of estimating 5-node and 6-node CIS class concentrations}}
Because the number of $k$-node CISes exponentially
increases with $k$,
it is computationally intensive to calculate the ground-truth
of $k$-node CIS classes' 
when $k \! \ge \!5$.
Nevertheless, we proceed to evaluate our methods based on a relatively
small graph Gnutella which has 6,299 nodes
and 20,776 edges for $k=5$ and $k=6$. Gnutella has $3.9\times 10^8$
five-node CISes and $1.7\times 10^{10}$ six-node CISes.
It takes almost one day to obtain all these subgraphs using the software
provided in Kashtan et al.~\cite{Kashtan2004}.
Fig.~\ref{fig:fiveandsixnode} shows NRMSEs of concentration estimates of one
five-node undirected CIS class and
one five-node undirected CIS class for Gnutella graph.
The five-node undirected CIS class we studied is
topologically equivalent to a five-node tree with depth one.
The true value of its concentration is 0.183 for Gnutella graph.
The results show that PSRW is nearly four times more accurate
than MHSRW and FANMOD.
The six-node undirected CIS class we studied is topologically
equivalent to a six-node tree with depth one.
The true value of its concentration is 0.0589 for the Gnutella graph.
The results show that PSRW is nearly twice as accurate as
MHSRW and FANMOD.

\subsection{\textbf{Time cost of sampling CISes}}
The time cost of sampling a $k$-node CIS consists of two parts:
1) computational time, and
2) the query response time.
We observe that the computation times increase with $k$ for PSRW, MHSRW, and FANMOD,
and the computation times are smaller than 0.1 second for $k\le 5$, which is usually
smaller than the query rate limits for querying a node imposed by OSNs.
Thus, we can easily find that PSRW is computationally more efficient than MHSRW,
since PSRW samples $k$-node CISes from graph $G^{(k-1)}$,
while MHSRW samples $k$-node CISes from graph $G^{(k)}$.
We compare the performances of different methods under the same time budget $T$.
We do simulations to evaluate the performance of different methods for two cases:
1) the graph is stored in a local database with near zero query delay,
and 2) the graph is stored in a remote database with 100 milliseconds query delay.
Fig.~\ref{fig:computationalcost} shows the NRMSEs of estimates of $\omega_2^{(3)}$ under $T$=200, 400, 600, 800, and 1,000 seconds for the Flickr graph.
The results show that PSRW is four and five times as accurate as the other methods for the same time budget $T$ for the local and remote databases respectively. The results for other graphs are similar, so we omit them here.

\begin{figure*}[htb]
\center
\subfigure[Epinons]{
\includegraphics[width=0.49\textwidth]{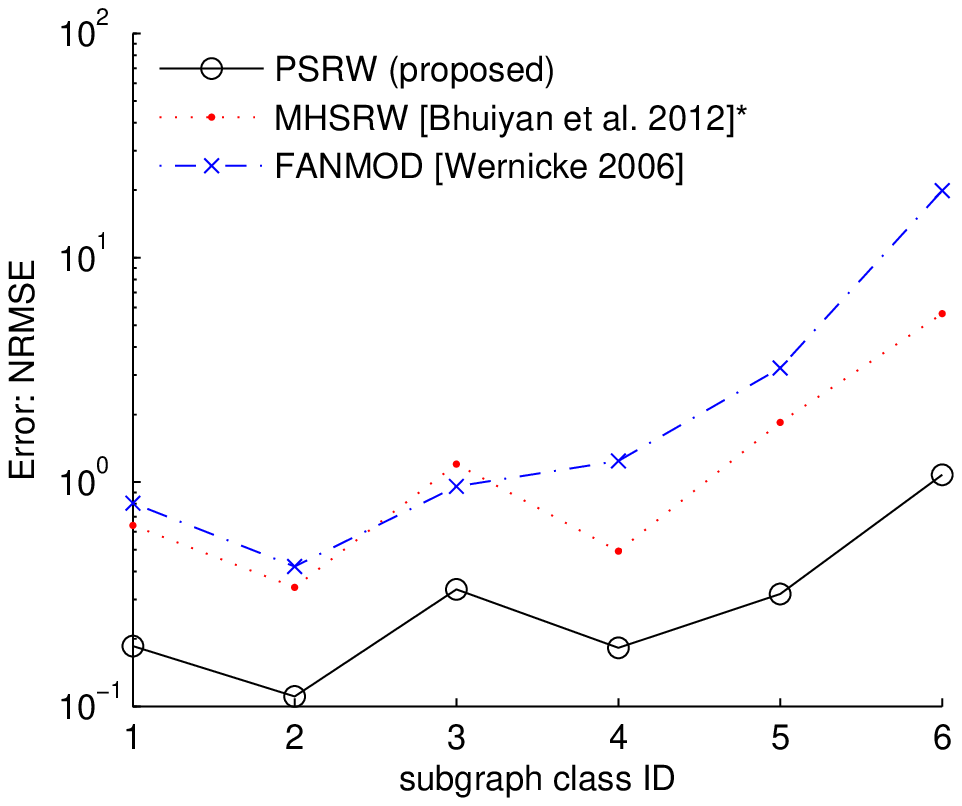}}
\subfigure[Slashdot]{
\includegraphics[width=0.49\textwidth]{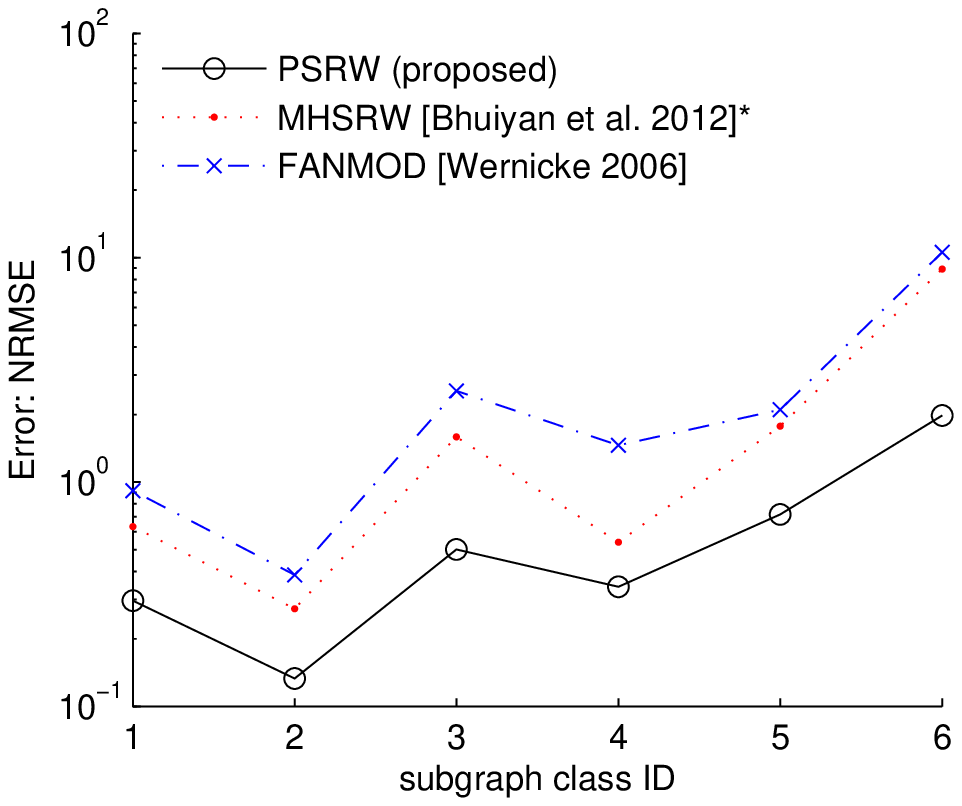}}
\subfigure[Gnutella]{
\includegraphics[width=0.49\textwidth]{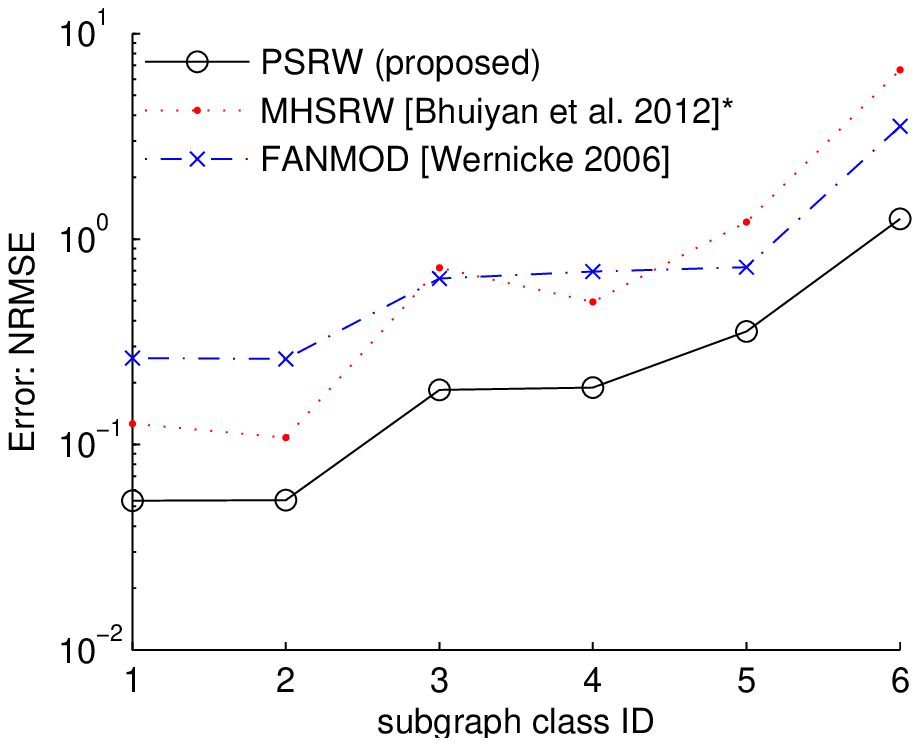}}
\caption{(Epinons and Slashdot) Compared NRMSEs of concentration estimates of 4-node undirected CIS classes for different methods under the same number of queries $B^*=2,000$.}\label{fig:undirectedfourclass}
\end{figure*}

\begin{figure*}[htb]
\center
\subfigure[5-node tree with depth one]{
\includegraphics[width=0.49\textwidth]{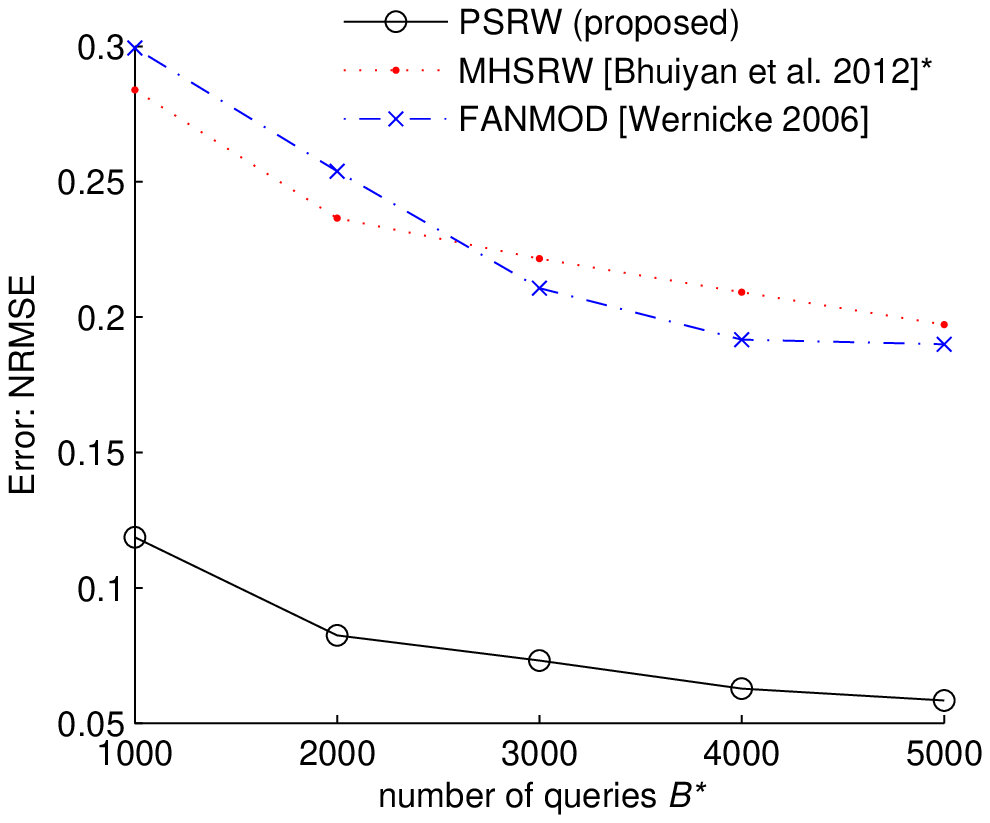}}
\subfigure[6-node tree with depth one]{
\includegraphics[width=0.49\textwidth]{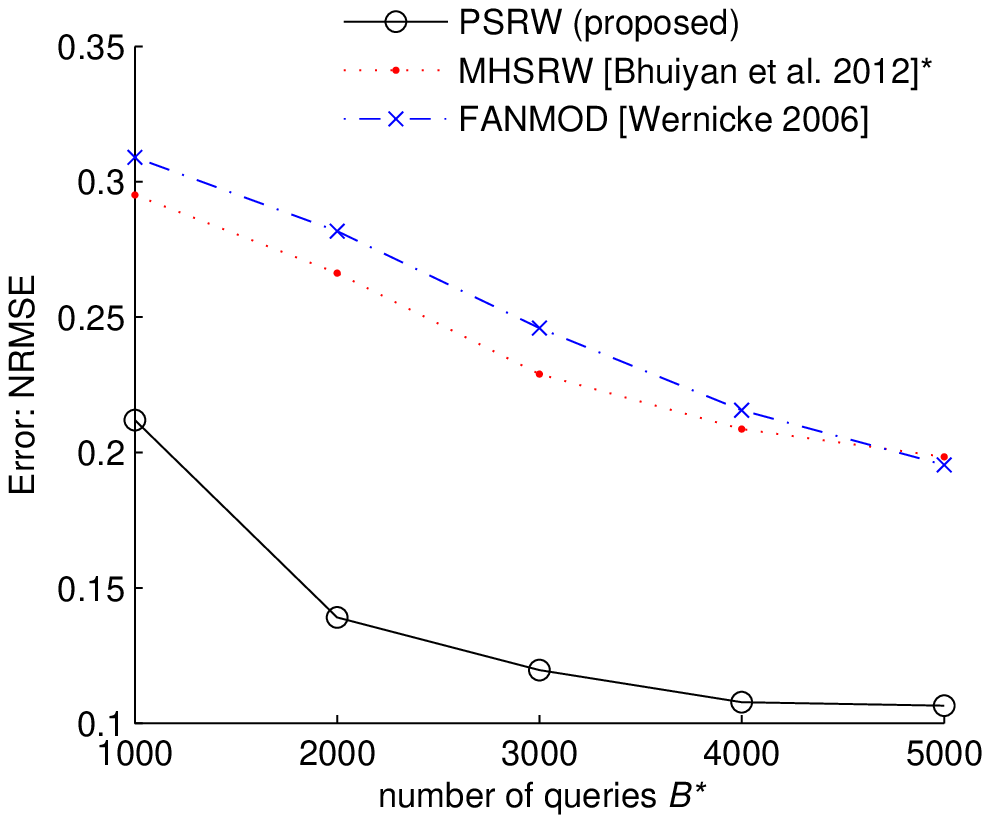}}
\caption{(Gnutella) Compared NRMSEs of concentration estimates of one 5-node undirected CIS class and one 6-node undirected CIS class for different methods.}\label{fig:fiveandsixnode}
\end{figure*}

\begin{figure*}[htb]
\center
\subfigure[Local database (near zero query delay)]{
\includegraphics[width=0.49\textwidth]{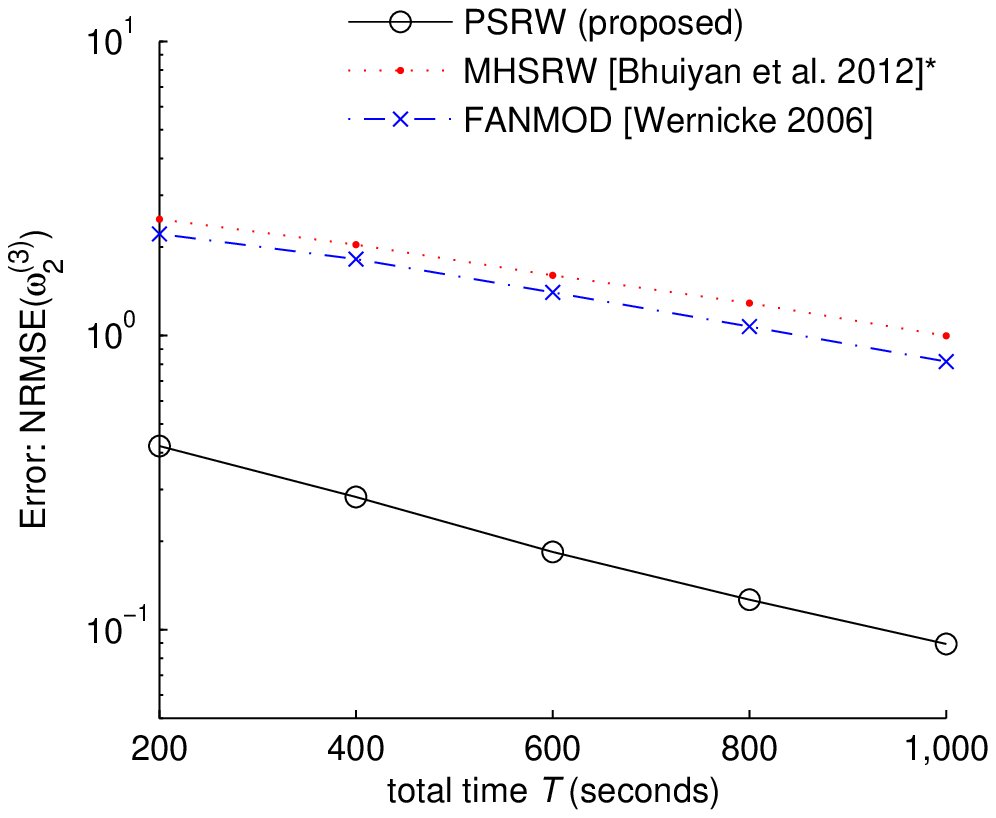}}
\subfigure[Remote database (100 milliseconds of query delay)]{
\includegraphics[width=0.49\textwidth]{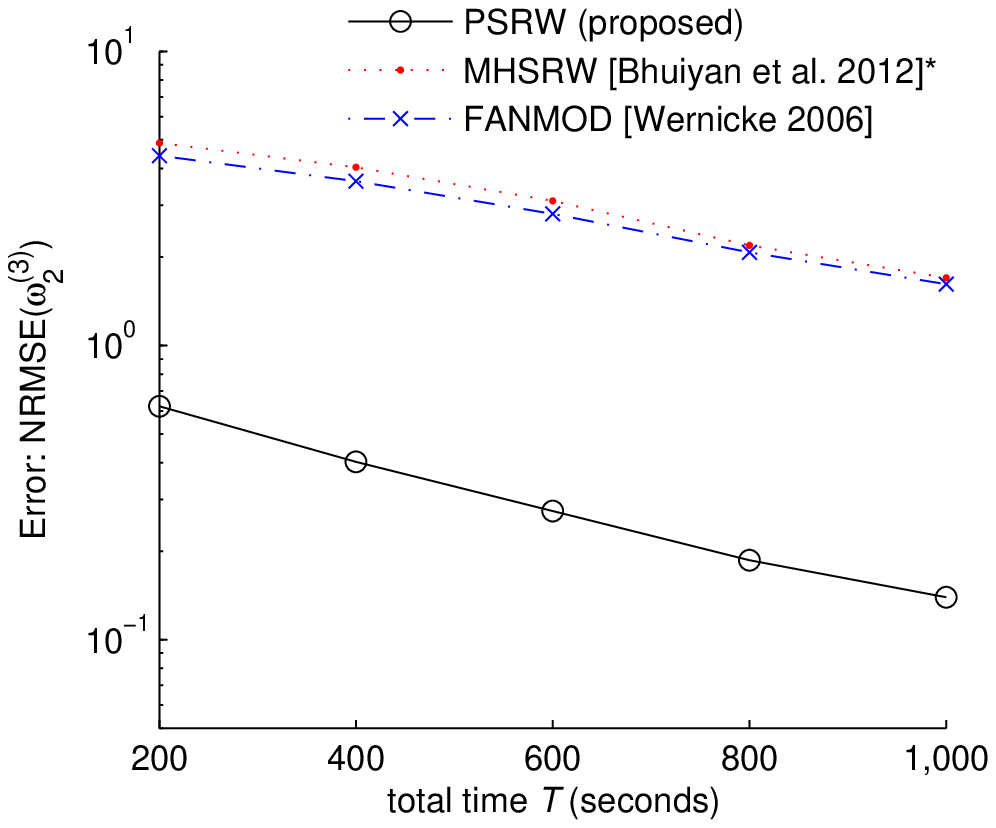}}
\caption{(Flickr) NRMSEs of error estimates of $\omega_2^{(3)}$ under for computer time $T$.}\label{fig:computationalcost}
\end{figure*}

\subsection{\textbf{Comparison with GUISE}}
Next, we evaluate the performance of our method MSS for the special case of simultaneous estimation of CISes $k= 3, 4, 5$ concentrations as in GUISE~\cite{Bhuiyan2012}.
Let $\omega^{(k)}=(\omega_1^{(k)},\ldots,\omega_{T_k}^{(k)})$ and $\hat{\omega}^{(k)}=(\hat{\omega}_1^{(k)},\ldots,\hat{\omega}_{T_k}^{(k)})$,
where $\omega_i^{(k)}$ is the concentration of subgraph class $C_i^{(k)}$ and $\hat{\omega}_i^{(k)}$  is an estimate of $\omega_i^{(k)}$.
We define the  root mean square error (RMSE) as:
\[
\text{RMSE}(\hat{\omega}^{(k)})=\sqrt{\text{E}[\sum_{i=1}^{T_k}(\hat{\omega}^{(k)}_i-\omega^{(k)}_i)^2]}, \qquad k=3,4, \text{ and } 5,
\]
which measures the error of the estimate $\hat{\omega}^{(k)}$
with respect to its true value $\omega^{(k)}$.
Note here the variable $\hat{\omega}^{(k)}$ of function $\text{RMSE}(\hat{\omega}^{(k)})$ is a \emph{vector}.
In our experiments, we average the estimates and calculate their
RMSEs over 1,000 runs.
Fig.~\ref{fig:CMPGuise} shows RMSEs of estimates of $\omega^{(3)}$, $\omega^{(4)}$, and $\omega^{(5)}$ for different methods under $B^*\!=\!3,000$ queries.
Besides MSS and GUISE, we also use PSRW to estimate $\omega^{(3)}$, $\omega^{(4)}$, and $\omega^{(5)}$ respectively.
For simplicity, we provide PSRW a budget of 1,000 queries for each value of $k$,
since it is hard to determine the optimal budget allocation for PSRW to jointly estimate $\omega^{(3)}$, $\omega^{(4)}$, and $\omega^{(5)}$,
The results show that MSS is more accurate than PSRW, and is nearly three times more accurate than GUISE.
\begin{figure*}[htb]
\center
\includegraphics[width=0.8\textwidth]{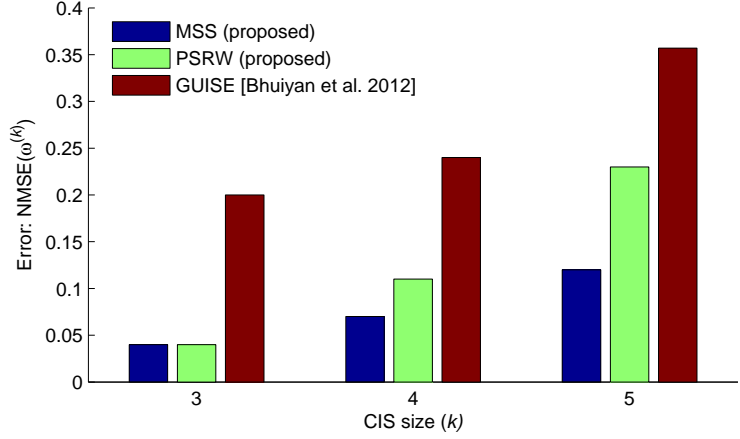}
\caption{(Gnutella) Compared errors of characterizing 3-node, 4-node, and 5-node undirected CIS classes simultaneously for different methods under the same number of queries $B^*=3,000$.}\label{fig:CMPGuise}
\end{figure*}

\section{Applications} \label{sec:application}
In this section, we apply our
methods to understand intrinsic properties of some large OSNs.  We conduct
experiments on Chinese OSNs Sina microblog\footnote{www.weibo.com} and Douban\footnote{www.douban.com}.  Sina
microblog is the most popular Chinese microblog service, and has many
features similar to Twitter.  It has more than 300 million registered users
as of February 2012.  Douban provides an exchange platform for reviews and
recommendations on movies, books, and music albums. It has approximately 6
million registered users as of 2009~\cite{ZhaoNetsci2011}.  Douban and Sina
microblog can be both modeled as directed graphs, where edges are formed by
users' following and follower relationships.
We conducted experiments in September 2012 on Sina microblog and Douban.
By using PSRW,
we sampled approximately 500,000 3-node CISes
from Sina microblog and Douban respectively.
Fig.~\ref{fig:exam_motif} (a) shows
the estimated concentrations of 3-node directed CIS classes.  It shows that closed
subgraph classes (classes 8--13) have much lower concentrations than unclosed
subgraph classes (classes 1--6), which indicates that Douban and Sina microblog
have a small fraction of closed triangles, and thus they have small clustering
coefficients, unlike the friendship relationship based OSNs such as
Facebook.  We observe that the concentration of subgraph class 7 is almost zero,
and is omitted from the concentration results shown in
Fig.~\ref{fig:exam_motif}. From Fig.~\ref{fig:3nodeclasses}, we observe that
subgraph class 7 is a directed circle of three nodes, which corresponds to three persons A, B, C,
with A following B, B following C, and C following A.
A concentration of zero might be explained by the asymmetry of following relationships,
where a following edge usually indicates statuses of the two end
users, e.g., an edge from a low-status user to high-status user such as a
celebrity.  Therefore, three users with different statuses are unlikely to
form a closed circle.

Next, we study the Z-scores of these subgraph classes.
The Z-score of each subgraph class $C_i^{(k)}$, $k>1$, is defined as
\begin{equation}\label{eq:zscore}
Z_i^{(k)}=\frac{\omega_i^{(k)}-\mu_{i}^{(k)}}{\sigma_{i}^{(k)}}, \qquad 1\le
i\le T_k, \end{equation}
where $\mu_{i}^{(k)}$ and $\sigma_{i}^{(k)}$ are the
mean and the standard deviation of the concentration of
$C_i^{(k)}$ for random
graphs with the same in-degree and out-degree sequence as $G_d$. Clearly the
Z-score of $C_i^{(k)}$ is a qualitative measure on the significance of
$C_i^{(k)}$~\cite{Milo2002}.

We propose a method to estimate $\mu_{i}^{(k)}$
and $\sigma_{i}^{(k)}$ as follows:
First, we use graph sample methods such as RW to estimate the joint
in-degree and out-degree distribution $\bphi=(\phi(i,j): i,j\ge 0)$,
where $\phi(i,j)$ is the fraction of nodes in $G_d$ with in degree $i$
and out degree $j$.
In essence, this is similar to the
problem of estimating node label densities as studied in our
previous work~\cite{Ribeiro2010}.
We use the configuration model~\cite{Molloy1995} to generate random
networks according to the estimated joint
in-degree and out-degree distribution $\hat\phi(i,j)$.
To generate a random graph, we first generate
$|V|$ nodes, and the in-degree and
out-degree of each node are
randomly selected according to $\hat\phi(i,j)$,
where the graph size $|V|$ can be estimated by
sampling methods proposed in~\cite{Katzir2011}.
We then use the configuration
model~\cite{Molloy1995} to generate a group of random graphs.
Algorithm~\ref{alg:randomgraph} describes the pseudo-code of our method for generating a random graph.
Finally, we compute the mean and standard deviation of the subgraph
class concentration based on randomly generated graphs.

\begin{algorithm}
\caption{Pseudo-code of random graph generation algorithm.}\label{alg:randomgraph}
\begin{algorithmic}[1]
\State \textbf{Step 1}: Assign each node $v$ with $d_I(v)$ incoming edge stubs
    (in-stubs) and $d_O(v)$ outgoing edge stubs (out-stubs).
\State \textbf{Step 2}: Pick an unconnected in-stub randomly from all nodes' in-stubs.
   Denote the associated node of selected in-stub as $v_i$.
\State \textbf{Step 3}: Pick an unconnected out-stub randomly from all nodes' out-stubs.
   Denote the associated node of selected out-stub as $v_o$.
   Repeat this step when $v_o=v_i$ or there already exists an edge from
   $v_i$ to $v_j$.
\State \textbf{Step 4}: Connect the selected in-stub and out-stub.
\State Repeat Step 2 to Step 4 until no unconnected in-stub or out-stub remains.
\end{algorithmic}
\end{algorithm}

Using the above method, we estimate the joint degree distribution based on
nearly one million unique nodes sampled by RW for Sina microblog and
Douban respectively, and then generate 1,000 random graphs to compute
the mean and the standard deviation of 3-node subgraph classes' concentrations,
which are used for estimating Z-scores shown in Eq.~(\ref{eq:zscore}).
Fig.~\ref{fig:exam_motif} (b) shows estimated Z-scores of
3-node directed CISes.  We find that subgraph classes 1 and 3 have
higher Z-scores in Sina microblog than Douban, where subgraph class 1 can be
viewed as a {\em listening type}, i.e., users follow many
celebrities, and subgraph class 3 can be viewed as a broadcast type,
i.e., celebrities have many fans.  This indicates that
Sina microblog acts more like a news media than an OSN,
which is similar to Twitter as observed in~\cite{Kwak2010}.
Subgraph class 6 has a higher Z-score in Douban than Sina microblog.
It may be because Douban is an interest-based network, where an edge
between two users with many common
interests is more likely to be symmetric than asymmetric.

\begin{figure}[htb] \center \subfigure[concentrations]{
\includegraphics[width=0.6\textwidth]{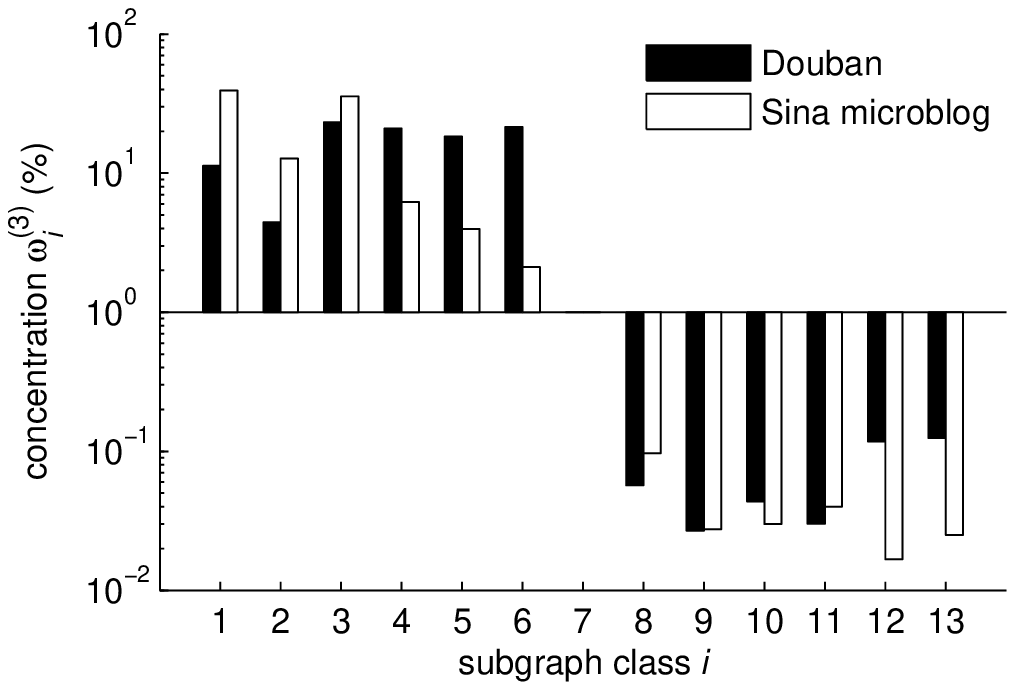}}
\subfigure[Z-scores]{
\includegraphics[width=0.6\textwidth]{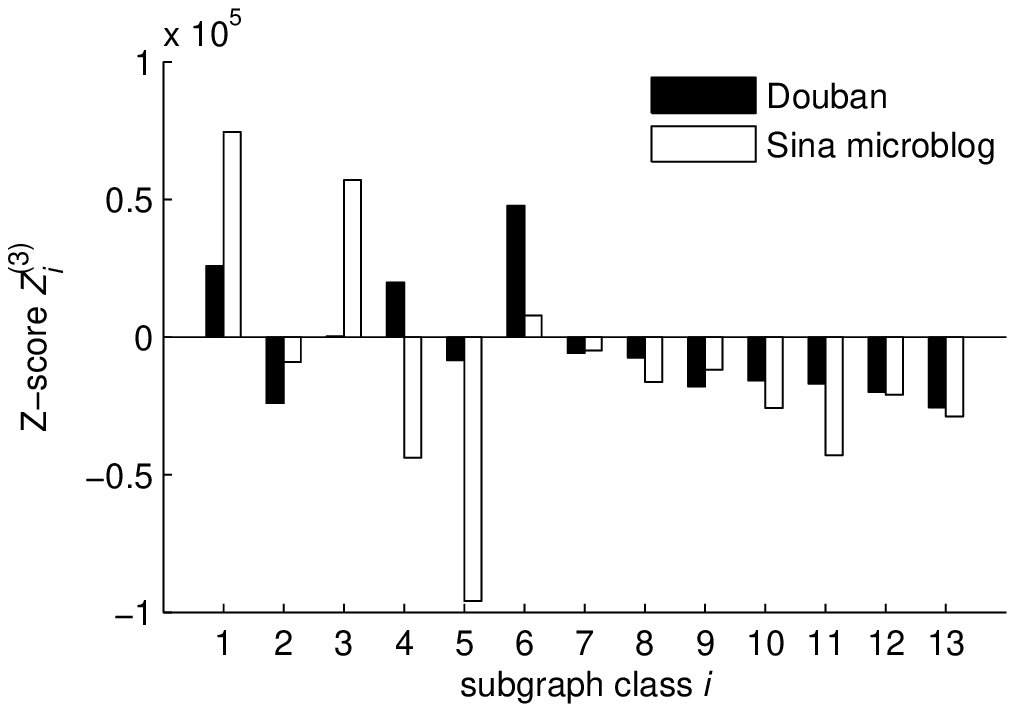}} \caption{Results of
real applications for all 3-node directed CISes.}\label{fig:exam_motif}
\end{figure}

\section{\textbf{Related Work}} \label{sec:related}
In this paper we aim to characterize small subgraphs in \emph{a single large graph},
which is a very different problem than that of estimating the number of subgraph patterns appearing in \emph{\textbf{a large set of graphs}} studied in~\cite{Hasan2009}.
Our problem can be directly solved by methods of enumerating all subgraphs of a specific size and type, such as triangle listing~\cite{ChuTKDD2012} and maximal clique enumeration~\cite{ChengTODS2011}.
There are several subgraph concentration computation methods for motif discovery using different subgraph enumeration and counting methods~\cite{ChenKDD2006,Kashani2009}.
However these methods need to process the whole graph and are computationally hard for large graphs.
Meanwhile most of these methods are difficult to combine with sampling techniques.
OmidiGenes et al.~\cite{OmidiGenes2009} proposed a subgraph enumeration and counting method using sampling.
However this method suffers from unknown sampling bias.
To estimate subgraph class concentrations, Kashtan et al.~\cite{Kashtan2004} proposed a connected subgraph sampling method using random edge sampling. However their method is computationally expensive when calculating the weight of each sampled subgraph, which is used for correcting bias introduced by edge sampling.
To address this drawback, Wernicke~\cite{Wernicke2006} proposed a new method named FANMOD based on enumerating subgraph trees to detect network motifs.
To sample a $k$-node CIS, their method needs to explore more than $k$ nodes, which is expensive when exploring graph topology via crawling.
Neither the method proposed in~\cite{Kashtan2004} nor~\cite{Wernicke2006} can be applied to detect motifs in OSNs without the complete knowledge of the graph topology,
since they rely on uniform edge sampling and uniform node sampling techniques respectively,
which may not be feasible because these sampling functions are not supported by most OSNs.

Similar to estimate subgraph class concentrations,
Bhuiyan et al.~\cite{Bhuiyan2012} propose a method GUISE for estimating 3-node, 4-node, and 5-node subgraph frequency distribution, that is, ($\frac{n_x}{N}: x$ is a 3-node, 4-node, or 5-node undirected and connected subgraph class), where $n_x$ be the number of undirected CISes in subgraph class $x$,
and $N$ is the total number of 3-node, 4-node, and 5-node undirected CISes.
GUISE builds a new graph $G_{mix}$, whose node set consists of all 3-node, 4-node, and 5-node CISes.
For a 3-node CIS, all 3-node and 4-node CISes having 2 and 3 nodes in common respectively are its neighbors in $G_{mix}$.
For a 4-node CIS, all 3-node, 4-node, and 5-node CISes having 3, 3, and 4 nodes in common respectively are its neighbors in $G_{mix}$.
For a 5-node CIS  all 4-node and 5-node CISes with 4 nodes in common are its neighbors in $G_{mix}$.
To estimate subgraph frequency distribution, GUISE performs a Metropolis-Hastings based sampling method
over $G_{mix}$. Hardiman and Katzir~\cite{HardimanandKatzir2013} propose random walk based sampling methods for estimating the network average and global clustering coefficients. Gjoka et al.~\cite{Gjoka2013} propose a uniform node sampling based method for estimating the clique (i.e., complete subgraph) size distribution. The methods in~\cite{HardimanandKatzir2013,Gjoka2013} are difficult to extend to measure concentrations of subgraph classes.

\section{\textbf{Conclusions}} \label{sec:conclusions}
In this paper we propose two random walk based sampling methods to estimate subgraph
class concentrations when the complete graph topology is not available.
The experimental results show that our methods PSRW and SRW only need to sample a very
small fraction of subgraphs to obtain an accurate and unbias estimate, and significant
reduces the number of samples required to achieve the same estimation accuracy of
state-of-the-art methods such as FANMOD.
Also, simulation results show that PSRW is much more accurate and computational efficient
than SRW.

\section*{\textbf{Appendix}}
\begin{lemma}\label{lemma:subgraphexisting} When a graph $G=(V,E)$ is connected,
 for each node $v\in V$, we can generate a $(k+1)$-node tree with a root $v$
that contains $min\{d(v),k\}$ neighbors,
where $d(v)$ is the degree of $v$ in graph $G$, and $1\le k\le |V|-1$.
\end{lemma}
\begin{proof}
One can use breadth-first search (BFS) to traverse $G$ starting from $v$,
then build a tree from the first $k$ nodes visited by BFS, where
$2\le k\le |V|$. This tree clearly contains
$min\{d(v),k\}$ neighbors of $v$ and the root node is $v$.
\end{proof}

\begin{lemma}[\cite{Roberts2004,Jones2004,LeeSigmetric2012}]\label{lemma:nodeestimator}
Let $G=(V, E)$ be connected and non-bipartite.
Let $u_j$ be the $j$-th node sampled by a RW on $G$,
where $1\le j\le B$ and $B$ be the number of samples.
Denote by $\bpi=(\pi_v, v\in V)$ the stationary distribution, where $\pi_v=\frac{d_v}{2|E|}$.
Then, for any function $f(v):V\rightarrow \mathbb{R}$,
where $\sum_{\forall v\in V} f(v)<\infty$, we have
\begin{equation*}
\lim_{B\rightarrow \infty} \frac{1}{B} \sum_{j=1}^B f(u_j) \xrightarrow{a.s.} \frac{1}{|V|} \sum_{\forall v\in V} f(v)\pi_v.
\end{equation*}
\end{lemma}

\begin{lemma}[\cite{Meyn2009,Ribeiro2010}]\label{lemma:edgeestimator}
Let $G=(V, E)$ be an undirected graph which is connected and non-bipartite.
Let $(u_j,v_j)$ ($1\le j\le B$) be the $j$-th edge sampled by a RW,
where $B$ is the number of sampled edges.
Denote function $f(u,v)\!:\!V\!\times\! V\!\rightarrow\! \mathbb{R}$.
Then, we have
\[
\lim_{B\rightarrow \infty} \frac{1}{B} \sum_{j=1}^B f(u_j, v_j) \xrightarrow{a.s.} \frac{1}{|E|} \sum_{\forall (u,v)\in E} f(u, v),
\]
for any function $f$ with $\sum_{\forall (u,v)\in E} f(u, v)<\infty$.
\hfill $\square$
\end{lemma}

\subsection{\textbf{Proof of Theorem~\ref{theorem:connected}}}
We use induction to prove $G^{(k)}$ is connected.

{\it Initial Step}. Since $G$ is connected, clearly there exists a path (edge sequence) between any two disconnected edges. Therefore $G^{(2)}$ is connected.
{\it Inductive Step}. Our inductive assumption is that $G^{(k)}$ is connected, $2\le k \le |V|-2$. We now prove that $G^{(k+1)}$ is also connected. For any two different CISes $x^{(k+1)}$ and $y^{(k+1)}$ in $C^{(k+1)}$, from Lemma~\ref{lemma:subgraphexisting} we can easily show that there exists a $k$-node CIS $x^{(k)}$ contained by  $x^{(k+1)}$, and a $k$-node CIS $y^{(k)}$ contained by  $y^{(k+1)}$.
When $x^{(k+1)}$ and $y^{(k+1)}$ are not connected, our inductive assumption shows that there exists a $k$-node CIS sequence $s_i^{(k)}$ ($1\le i\le l$) in graph $G^{(k)}$, where $s_1^{(k)}$ connects to $x^{(k)}$, $s_l^{(k)}$ connects to $y^{(k)}$, and two adjacent $k$-node CIS $s_i^{(k)}$ and $s_{i+1}^{(k)}$ are connected, where $2\le i< l$.
Denote by $s_1^{(k+1)}$ the $(k+1)$-node CIS consisting of $k+1$ different nodes appearing in $s_1^{(k)}$ and $x^{(k)}$,
$s_{l+1}^{(k+1)}$ the $(k+1)$-node CIS consisting of $k+1$ different nodes appearing in $s_l^{(k)}$ and $y^{(k)}$,
and $s_{i}^{(k+1)}$ the $(k+1)$-node CIS consisting of $k+1$ different nodes appearing in $s_i^{(k)}$ and $s_{i+1}^{(k)}$,
where $2\le i< l$.
In graph $G^{(k+1)}$, we can easily find that $s_1^{(k+1)}$ connects to $x^{(k+1)}$, $s_{l+1}^{(k+1)}$ connects to $y^{(k+1)}$, and two adjacent $(k+1)$-node CISes $s_i^{(k+1)}$ and $s_{i+1}^{(k+1)}$ ($1\le i\le l$) are connected. This shows that there exists a path between any two disconnected nodes ($(k+1)$-node CISes) in $G^{(k+1)}$. Therefore graph $G^{(k+1)}$ is connected.

\subsection{\textbf{Proof of Theorem~\ref{theorem:bipartitedegree}}}
Denote by $v$ the node with degree larger than two. Lemma~\ref{lemma:subgraphexisting} indicates that there exists a $k$-node tree $t$ with root $v$ which contains at least three neighbors of $v$, where $4\le k\le |V|$. We easily find that $t$ has at least three leaves. Since $t$ is still connected after we remove any leaf, there exist at least three different $(k-1)$-node CISes consisting of $k-1$ nodes in $t$ obtained by removing one leaf of $t$, and these CISes are connected to each other in graph $G^{(k-1)}$. Similarly there exist at least three $(k-2)$-node CISes consisting of $k-2$ nodes in $t$ by excluding two leaves of $t$, which are connected to each other in $G^{(k-2)}$. Therefore, each $G^{(k)}$ ($2\le k<|V|$) is non-bipartite since it has at least one odd length loop.
When $G$ has no node with degree larger than two, since $G$ is connected and non-bipartite, we can easily show that $G$ is a $|V|$-node circle and $|V|$ is odd. For each node $v\in V$, we can generate a $k$-node CIS consisting of $v$ and $k-1$ nodes close to $v$ in clockwise direction, where $2\le k<|V|$. Finally there are $|V|$ different $k$-node CISes, and they form an odd length loop in graph $G^{(k)}$. Therefore $G^{(k)}$ is non-bipartite.

\subsection{\textbf{Proof of Theorem~\ref{theorem:nodeestimatorunbiased}}}
SRW can be viewed as a regular RW over graph $G^{(k)}$, $2\le k<|V|$.
For each $\omega_i^{(k)}$, $1\le i\le T_k$,  we then obtain following equations from
Lemma~\ref{lemma:nodeestimator} and Theorem~\ref{theorem:srwstationary} for
non-bipartite and connected $G^{(k)}$,
\begin{equation*}
\begin{split}
& \hspace{-0.2in}
 \lim_{B\rightarrow \infty} \frac{1}{B} \sum_{j=1}^B \frac{\mathbf{1}(C(s_j)= C_i^{(k)})}{d^{(k)}(s_j)}\\
&\xrightarrow{a.s.} \frac{1}{|C^{(k)}|} \sum_{\forall s\in C^{(k)}} \frac{\mathbf{1}(C(s)= C_i^{(k)})}{d^{(k)}(s)} \pi^{(k)}(s)\\
&=\frac{1}{|C^{(k)}| \sum_{t\in C^{(k)}} d^{(k)} (t)} \sum_{\forall s\in C^{(k)}} \mathbf{1}(C(s)= C_i^{(k)})\\
&=\frac{\omega_i^{(k)}}{\sum_{t\in C^{(k)}} d^{(k)} (t)}.
\end{split}
\end{equation*}
Similarly, we have
\[
\lim_{B\rightarrow \infty} \frac{1}{B}\sum_{j=1}^B \frac{1}{d^{(k)}(s_j)}\xrightarrow{a.s.}\frac{1}{\sum_{t\in C^{(k)}} d^{(k)} (t)}.
\]
Thus, we can easily find that $\hat\omega_i^{(k)}$ ($1\le i\le T_k$) is an asymptotically unbiased estimator of $\omega_i^{(k)}$.

\subsection{\textbf{Proof of Theorem~\ref{theorem:edgeestimatorunbiased}}}
To estimate $\tilde\omega_i^{(k)}$, $1\le i\le T_{k}$, $2\le k< |V|$,
PSRW can be viewed as a regular RW over the graph $G^{(k-1)}$.
Denote $s^*_{(u,v)}$ as the $k$-node CIS generated by $(u,v) \in R^{(k-1)}$, an edge in $G^{(k-1)}$, where $u, v\in C^{(k-1)}$ are $(k-1)$-node CISes.
For each $\omega_i^{(k)}$, $1\le i\le T_k$,  we then obtain
following equations from Lemma~\ref{lemma:edgeestimator},
\begin{equation*}
\begin{split}
&\lim_{B\rightarrow \infty} \frac{1}{B-1}\sum_{j=1}^{B-1} \frac{\mathbf{1}(C(s^*_j)= C_i^{(k)})}{I^{(k-1)}(s^*_j)\left(I^{(k-1)}(s^*_j)-1\right)}\\
&\xrightarrow{a.s.}\frac{1}{|R^{(k-1)}|} \sum_{\forall (u,v) \in R^{(k-1)}} \frac{\mathbf{1}(C(s^*_{(u,v)})= C_i^{(k)})}{I^{(k-1)}(s^*_{(u,v)})\left(I^{(k-1)}(s^*_{(u,v)})-1\right)}\\
&=\frac{1}{2|R^{(k-1)}|} \sum_{\forall s\in C^{(k)}} \mathbf{1}(C(s)= C_i^{(k)})\\
&=\frac{\omega_i^{(k)}|C^{(k)}|}{2|R^{(k-1)}|}.
\end{split}
\end{equation*}
The last equation holds because
the $k$-node CIS $s$
is generated by $\frac{\left(I^{(k-1)}(s)\right)\left(I^{(k-1)}(s)-1\right)}{2}$
edges in $R^{(k-1)}$.
Similarly, we have
\[
\lim_{B\rightarrow \infty} \frac{\sum_{j=1}^{B-1} \frac{1}{I^{(k-1)}(s^*_j)\left(I^{(k-1)}(s^*_j)-1\right)}}{B-1} \xrightarrow{a.s.}\frac{|C^{(k)}|}{2|R^{(k-1)}|}.
\]
Thus, we can easily find that $\tilde\omega_i^{(k)}$ ($1\le i\le T_{k}$) is an asymptotically unbiased estimator of $\omega_i^{(k)}$.

\subsection{\textbf{Proof of Theorem~\ref{theorem:reduceestimatorunbiased}}}
For each $\omega_i^{(k-1)}$, $1\le i\le T_{k-1}$,  we obtain following equations from
Lemma~\ref{lemma:nodeestimator} and Theorem~\ref{theorem:srwstationary} for
non-bipartite and connected $G^{(k)}$,
\begin{equation*}
\begin{split}
&\lim_{B\rightarrow \infty} \frac{1}{B} \sum_{j=1}^B \frac{1}{d^{(k)}(s_j)} \sum_{s'\in C^{(k-1)}(s_j)} \frac{\mathbf{1}(C^{(k-1)}(s')= C_i^{(k-1)})}{|O^{(k)}(s')|}\\
&\xrightarrow{a.s.}\frac{1}{|C^{(k)}|} \sum_{\forall s\in C^{(k)}} \frac{\pi^{(k)}(s)}{d^{(k)}(s)} \sum_{s'\in C^{(k-1)}(s)} \frac{\mathbf{1}(C^{(k-1)}(s')= C_i^{(k-1)})}{|O^{(k)}(s')|}\\
&=\frac{1}{|C^{(k)}| \sum_{t\in C^{(k)}} d^{(k)} (t)} \sum_{\forall s\in C^{(k)}} \sum_{s'\in C^{(k-1)}(s)} \frac{\mathbf{1}(C^{(k-1)}(s')= C_i^{(k-1)})}{|O^{(k)}(s')|}\\
&=\frac{1}{|C^{(k)}| \sum_{t\in C^{(k)}} d^{(k)} (t)} \sum_{\forall s'\in C^{(k-1)}} \sum_{s\in O^{(k)}(s')} \frac{\mathbf{1}(C^{(k-1)}(s')= C_i^{(k-1)})}{|O^{(k)}(s')|}\\
&=\frac{\omega_i^{(k-1)}|C^{(k-1)}|}{|C^{(k)}| \sum_{t\in C^{(k)}} d^{(k)} (t)}.
\end{split}
\end{equation*}
Similarly, we have
\[
\lim_{B\rightarrow \infty} \frac{\sum_{j=1}^B
\frac{1}{d^{(k)}(s_j)} \sum_{s'\in C^{(k-1)}(s_j)} \frac{1}{|O^{(k)}(s')|}}{B} \xrightarrow{a.s.}\frac{|C^{(k-1)}|}{|C^{(k)}| \sum_{t\in C^{(k)}} d^{(k)} (t)}.
\]
Thus, we can easily find that $\breve\omega_i^{(k-1)}$ ($1\le i\le T_{k-1}$) is an asymptotically unbiased estimator of $\omega_i^{(k-1)}$.

\bibliographystyle{ACM-Reference-Format-Journals}

\end{document}